\documentclass{amsart}[11pt]
\usepackage{amsmath, amsfonts, amsthm, amssymb,mathtools}
\usepackage{subfigure}
\usepackage{stmaryrd}
\usepackage{verbatim}
\usepackage{hyperref}
\usepackage{color}
\usepackage{ulem}
\usepackage[dvips]{graphicx}
\usepackage{enumerate}
\numberwithin{equation}{section}
\newtheorem{theorem}{Theorem}[section]
\newtheorem{corollary}[theorem]{Corollary}
\newtheorem{lemma}[theorem]{Lemma}

\theoremstyle{definition}
\newtheorem{definition}[theorem]{Definition}
\newtheorem{remark}[theorem]{Remark}

\newcommand{\ind}{1\hspace{-2.1mm}{1}} 
\newcommand{\I}{\mathtt{i}}
\newcommand{\RR}{\mathbb{R}}
\newcommand{\PP}{\mathbb{P}}
\newcommand{\D}{\mathrm{d}}

\newcommand{\E}{\mathrm{e}}

\newcommand{\sgn}{\mathrm{sgn}}
\newcommand{\Nn}{\mathcal{N}}
\newcommand{\Ff}{\mathcal{F}}

\begin{document}

\title{The Small-Maturity Heston Forward Smile}
\author{Antoine Jacquier and Patrick Roome}
\address{Department of Mathematics, Imperial College London}
\email{a.jacquier@imperial.ac.uk, p.roome11@imperial.ac.uk}
\date{\today}
\thanks{
The authors would like to thank all the participants of the ETH-Imperial workshop on mathematical finance, March 2013 as well as Peter Tankov, Claude Martini, Jim Gatheral and Leif Andersen for useful discussions.
They are also indebted to the anonymous referees for helpful comments.
}

\maketitle
\begin{abstract}
In this paper we investigate the asymptotics of forward-start options and the forward implied volatility smile in the Heston model as the maturity approaches zero.
We prove that the forward smile for out-of-the-money options explodes and compute a closed-form high-order expansion detailing the rate of the explosion.
Furthermore the result shows that the square-root behaviour of the variance process induces a singularity such that for certain parameter configurations one cannot obtain high-order out-of-the-money forward smile asymptotics.
In the at-the-money case a separate model-independent analysis shows that the small-maturity limit is well defined for any It\^o diffusion.
The proofs rely on the theory of sharp large deviations (and refinements) and
incidentally we provide an example of degenerate large deviations behaviour.
\end{abstract}

\section{Introduction}

Consider an asset price process $\left(\E^{X_t}\right)_{t\geq0}$ with $X_0=0$, paying no dividend, 
defined on a complete filtered probability space $(\Omega,\mathcal{F},(\mathcal{F}_t)_{t\geq0},\mathbb{P})$ 
with a given risk-neutral measure $\mathbb{P}$, and assume that interest rates are zero. 
In the Black-Scholes-Merton (BSM) model, the dynamics of the logarithm of the asset price are given by
$\D X_t=-\frac{1}{2}\sigma^2  \D t +\sigma \D W_t$,
where $\sigma>0$ represents the instantaneous volatility and $W$ is a standard Brownian motion.
The no-arbitrage price of the call option at time zero is then given by the famous BSM formula~\cite{BS73,M73}:
$C_{\textrm{BS}}(\tau,k,\sigma) := \mathbb{E}\left(\E^{X_\tau}-\E^k\right)_+
=\Nn\left(d_+\right)-\E^k\Nn\left(d_-\right)$,
with $d_\pm:=-\frac{k}{\sigma\sqrt{\tau}}\pm\frac{1}{2}\sigma\sqrt{\tau}$,
where $\Nn$ is the standard normal distribution function. 
For a given market price $C^{\textrm{obs}}(\tau,k)$ of the option at strike~$\E^k$ and maturity~$\tau$,
 the spot implied volatility $\sigma_{\tau}(k)$ is the unique solution to 
the equation $C^{\textrm{obs}}(\tau,k)=C_{\textrm{BS}}(\tau,k,\sigma_{\tau}(k))$.

For any~$t,\tau>0$ and $k\in\mathbb{R}$, we define as in~\cite{B04,VL}
a  forward-start option with forward-start date~$t$, maturity~$\tau$ and strike~$\E^k$ 
as a European option with payoff 
$\left(\E^{X_{\tau}^{(t)}}-\E^k\right)^+$ where we define $X_{\tau}^{(t)}:=X_{t+\tau}-X_t$ pathwise. 
By the stationary increment property, its value is simply $C_{\textrm{BS}}(\tau,k,\sigma)$ in the BSM model. 
For a given market price $\overline{C}^{\textrm{obs}}(t,\tau,k)$ of the option at strike~$\E^k$, 
forward-start date~$t$ and maturity~$\tau$,
the forward implied volatility smile~$\sigma_{t,\tau}(k)$ is then defined (see also~\cite{B04,Hong})
 as the unique solution to 
$\overline{C}^{\textrm{obs}}(t,\tau,k)=C_{\textrm{BS}}(\tau,k,\sigma_{t,\tau}(k))$. 
The forward smile is a generalisation of the spot implied volatility smile, and the two are equal when $t=0$. 

Asymptotics of the spot implied volatility surface have received a large amount of attention over the past decade. 
These results have helped shape calibration methodologies based on arbitrage-free approximation of the spot smile in a large variety of models.  
Small-maturity asymptotics have been studied by Berestycki-Busca-Florent ~\cite{BBF} using PDE methods for continuous time diffusions and by
Henry-Labord\`ere ~\cite{PHL} using heat kernel expansions.
Forde et al.~\cite{FJL11} and Jacquier et al.~\cite{JMKR} derived small and large-maturity asymptotics in the Heston model and affine stochastic volatility models with jumps using large deviations and saddlepoint methods.
Wing asymptotics (as $|k|\nearrow\infty$) have been studied by Lee ~\cite{Lee} and further extended by Benaim and Friz~\cite{BF1, BF2} and in~\cite{DFJV, FGGS, Guli}.
Fouque et al.~\cite{Fouque} have used  perturbation techniques in order to study slow and fast mean-reverting stochastic volatility models.
Small-noise expansions have been studied by Osajima~\cite{O07} and Takahashi~\cite{KT04} using Watanabe expansions and in
Deuschel et al.~\cite{DFJV} using Laplace method on Wiener space.
See~\cite{GuliBook,Zeliade} for a general overview of implied volatility asymptotics in stochastic volatility models.
In exponential L\'evy models implied volatilities of out-of-the-money options explode as the maturity tends to zero, while the implied volatility for at-the-money options converges to the volatility of the diffusion component as the maturity tends to zero.
Small-maturity asymptotics for models with jumps (including L\'evy processes) 
have been investigated in~\cite{Alos, AndLipton, Tank, Nutz, MijTankov, FL}.

However, these asymptotics do not provide any information on the forward smile or forward-start-based payoffs, such as cliquets and forward-start options and the literature on this topic is sparse. 
Glasserman and Wu~\cite{GW10} introduced different notions of forward volatilities for forecasting purposes,
Keller-Ressel~\cite{KR} studies a very specific type of asymptotic (when the forward-start date~$t$ becomes large), and empirical results have been carried out in~\cite{B04,G06,B02}. 
In~\cite{EM09} the authors empirically studied the forward smile in Sato models and ran comparisons with a suite of models including Heston and local volatility models for forward smile sensitive products such as cliquets.
Recently, small and large-maturity forward smile asymptotics were derived in~\cite{JR12} for a general class of models including the Heston model.
However, the results in~\cite{JR12} only apply to the so-called diagonal small-maturity regime, 
i.e. the behaviour (as $\varepsilon$ tends to zero) of the process~$(X_{\varepsilon \tau}^{(\varepsilon t)})_{\varepsilon\geq0}$.
The conjecture, stated there, is that for fixed $t>0$ the Heston forward smile 
(corresponding to~$X_{\tau}^{(t)}$) explodes to infinity (except at-the-money) as $\tau$ tends to zero.

In this paper we confirm this conjecture and give a high-order expansion for the forward smile.
The main result (Theorem~\ref{theorem:fwdsmilesmalltime}) is that the small-maturity Heston forward smile explodes according to the following asymptotic: 
$\sigma_{t,\tau}^2(k)=v_0(k,t)\tau^{-1/2}+v_{1}(k,t)\tau^{-1/4}+o\left(\tau^{-1/4}\right)$
for $k\in\mathbb{R}^*$ and $t>0$ as $\tau$ tends to zero.
Here $v_0(\cdot,t)$ and $v_1(\cdot,t)$ are even continuous functions (over $\mathbb{R}$) with $v_0(0,t)=v_1(0,t)=0$ and independent of the Heston correlation.
In the at-the-money case ($k=0$) a separate model-independent analysis (Lemma~\ref{lem:modelindatm} and Theorem~\ref{theorem:fwdatmsmalltime}) shows that the small-maturity limit is well defined and  
$\lim_{ \tau \searrow 0}\sigma_{t,\tau}(0)=\mathbb{E}(\sqrt{V_t})$ holds for any well-behaved diffusion where~$V_t$ is the instantaneous variance at time~$t$.
This exploding nature is consistent with empirical observations in~\cite{B02} and the diagonal small-maturity asymptotic from~\cite{JR12}.

The paper is structured as follows. 
In Section~\ref{sec:fwdtimescales} we introduce the notion of a forward time-scale and characterise it in the Heston model. 
In Section~\ref{sec:smfwdstartasymp} we state the main result on small-maturity asymptotics of forward-start options in the Heston model. 
Section~\ref{sec:smfwdsmileasymp} tackles the forward implied volatility asymptotics: Section~\ref{sec:otmasymp} translates the results of Section~\ref{sec:smfwdstartasymp} into out-of-the-money forward smile asymptotics, and
Section~\ref{sec:ATMfwdsmile} presents a model-independent result for the at-the-money forward implied volatility.
Section~\ref{sec:smhnumerics} provides numerical evidence supporting the asymptotics derived in the paper and the main proofs are gathered in Section~\ref{sec:proofssmalltau}.

\textbf{Notations: } $\mathbb{E}$ shall always denote expectation under a risk-neutral measure~$\PP$ given a priori.
We shall refer to the standard (as opposed to the forward) implied volatility as the spot smile and denote it $\sigma_\tau$.
The forward implied volatility will be denoted $\sigma_{t,\tau}$ as above. In this paper we will always assume that forward-start date is greater than zero ($t>0$) unless otherwise stated.

\section{Forward time-scales}\label{sec:fwdtimescales}

In this section we introduce the notion of a forward time-scale and characterise it in the Heston model.
In the BSM model the time-scale is given by $h(t)\equiv t$, 
which is related to the quadratic variation of the driving Brownian motion.
For diffusions (such as Heston) the (spot) time-scale is the same, which implies that the spot smile
has a finite (non-zero) small-maturity limit.
In the forward case, this however no longer remains true. 
Stochastic volatility models (eg. Heston) exhibit different time-scales to the BSM model leading to different asymptotic regimes for the forward smile relative to the spot smile.
As we will show below (Lemma~\ref{Lemma:HestFwdTimeScale}), 
all re-scalings of the Heston model lead to limiting logarithmic moment generating functions (lmgf's) that are all zero on their domains of definition. 
But the forward time-scale is the only choice that leads to the limiting lmgf being zero on a bounded domain. 
This is one of the key properties that allows us to derive sharp large deviation results even though at first sight this zero limit appears trivial and non-consequential. 
We define the forward lmgf by 
\begin{equation}\label{eq:MGFFwd}
\Lambda_{\tau}^{(t)}(u,a):=a\log\mathbb{E}\left(\E^{u X_{\tau}^{(t)}/a}\right),
\quad\text{for all }u\in\mathcal{D}_{t,\tau},
\end{equation}
where 
$\mathcal{D}_{t,\tau}:=\{u\in\mathbb{R}:|\Lambda_{\tau}^{(t)}(u,a)|<\infty\}$.
With this definition the domain $\mathcal{D}_{t,\tau}$ will depend on $a$, but it will be clear from the context which choice of $a$ we are using.

\begin{definition}\label{def:fwdtimescale}
We define a (small-maturity) \textit{forward time-scale} 
as a continuous function $h:\mathbb{R}_{+}\to\mathbb{R}_{+}$ such that $\lim_{\tau\searrow 0}h(\tau)=0$ 
and $\Lambda(u):=\lim_{\tau\searrow 0}\Lambda_{\tau}^{(t)}(u,{h(\tau)})$ produces a non-trivial pointwise limit. 
We shall say that a (pointwise) limit is trivial if it is null on~$\mathbb{R}$ 
or null at the origin and infinite on~$\mathbb{R}^*$.
\end{definition}

\begin{remark} \
\begin{enumerate}[(i)]
\item
In the BSM model the forward time-scale is $h(\tau)\equiv\tau$. 
\item
A forward time-scale may not exist for a model. 
For example, consider exponential L\'evy models with bounded domain for the L\'evy exponent.
The only non-trivial limit occurs when $h\equiv 1$, which does not satisfy Definition~\ref{def:fwdtimescale} and so a forward-time scale does not exist. 
\item
If $(X_{\tau}^{(t)})_{\tau\geq0}$ satisfies a large deviations principle~\cite[Section 1.2]{DZ93} with speed $h$ and assuming further some tail condition (see~\cite[Theorem 4.3.1]{DZ93}), then $h$ is the forward time-scale for the model by Varadhan's lemma. 
\item
Diffusion models have the same spot time-scale ($t=0$) as the BSM model, namely $h(\tau)\equiv \tau$ (see for example~\cite{BBF}).
This is not necessarily true in the forward case as we will shortly see.
\end{enumerate}
\end{remark}

In the Heston model, the (log) stock price process is the unique strong solution to the following SDEs:
\begin{equation}\label{eq:Heston}
\begin{array}{rll}
\D X_t & = \displaystyle -\frac{1}{2}V_t\D t+ \sqrt{V_t}\D W_t, \quad & X_0=0,\\
\D V_t & = \kappa\left(\theta-V_t\right)\D t+\xi\sqrt{V_t}\D Z_t, \quad & V_0=v>0,\\
\D\left\langle W,Z\right\rangle_t & = \rho \D t,
\end{array}
\end{equation}
with $\kappa>0$, $\xi>0$, $\theta>0$, $|\rho|<1$ and $(W_t)_{t\geq0}$ and $(Z_t)_{t\geq0}$ are two standard Brownian motions.
The Feller SDE for the variance process has a unique strong solution 
by the Yamada-Watanabe conditions~\cite[Proposition 2.13, page 291]{KS97}). 
The $X$ process is a stochastic integral of the $V$ process and is therefore well defined.  
The Feller condition, $2\kappa\theta\geq\xi^2$, ensures that the origin is unattainable. 
Otherwise the origin is regular (hence attainable) and strongly reflecting
(see~\cite[Chapter 15]{KT81}). 
We do not require the Feller condition in our analysis since we work with the forward lmgf of $X$ which is always well defined.
In order to characterise the Heston forward time-scale we require the following lemma, proved in Section~\ref{sec:hestfwdtime}.
\begin{lemma} \label{Lemma:HestFwdTimeScale}
Let $h:\RR_+\to\RR_+$ be a continuous function such that $\lim_{\tau\searrow 0}h(\tau)=0$ and $a\in\mathbb{R}_+^*$. 
The following limits hold for the Heston forward lmgf  as $\tau$ tends to zero with $\beta_t$ defined in~\eqref{eq:DGammaBeta}: 
\begin{enumerate}[(i)]
\item
If $h(\tau)\equiv a\sqrt{\tau}$ then $
  \lim_{\tau\searrow 0}\Lambda_{\tau}^{(t)}(u,{h(\tau)})=0,
$
for all $|u|<a/\sqrt{\beta_t}$ and is infinite otherwise;
\item
if $\sqrt{\tau}/h(\tau)\nearrow\infty$ then $
  \lim_{\tau\searrow 0}\Lambda_{\tau}^{(t)}(u,{h(\tau)})=0,
$
for $u=0$ and is infinite otherwise;
\item
if $\sqrt{\tau}/h(\tau)\searrow 0$ then
$\lim_{\tau\searrow 0}\Lambda_{\tau}^{(t)}(u,{h(\tau)})=0$,
for all $u\in\mathbb{R}$.
\end{enumerate}
\end{lemma}

As it turns out all limits are zero on their domains of definition, but using $h(\tau)\equiv\sqrt{\tau}$ produces the only 
(up to a constant multiplicative factor) non-trivial zero limit. 
It follows that $\tau\mapsto\sqrt{\tau}$ is the Heston forward time-scale. 
Let now $\Lambda:\mathcal{D}_{\Lambda}=(-1/\sqrt{\beta_t},1/\sqrt{\beta_t})\to\mathbb{R}$ 
be the pointwise limit 
(with $\beta_t:=\xi^2\left(1-\E^{-\kappa t}\right)/(4\kappa)$) from Lemma~\ref{Lemma:HestFwdTimeScale}, i.e.
satisfying $\Lambda(u)=0$ for~$u\in\mathcal{D}_{\Lambda}$ and infinity otherwise. 
Further we define the function $\Lambda^*:\mathbb{R}\to\mathbb{R}_+$ as the Fenchel-Legendre transform of $\Lambda$:
\begin{align}\label{eq:UStarDefinition}
\Lambda^*(k):=\sup_{u\in\mathcal{D}_{\Lambda}}\left\{uk-\Lambda(u)\right\},
\qquad\text{for all }k\in\mathbb{R}.
\end{align}

\begin{lemma}\label{lemma:U*Characterisation}
The function $\Lambda^*$ defined in~\eqref{eq:UStarDefinition} is characterised explicitly as 
$
\Lambda^*(k)=|k|/\sqrt{\beta_t}
$
for all $k\in\mathbb{R}$.
\end{lemma}
\begin{proof}
Clearly $\Lambda^*(0)=0$. Now suppose that $k>0$. 
Then $\Lambda^*(k)=\sup_{u\in\mathcal{D}_{\Lambda}}\left\{uk\right\}=k/\sqrt{\beta_t}$. A similar result holds for $k<0$ and the result follows.
\end{proof}

\section{Small-maturity forward-start option asymptotics}\label{sec:smfwdstartasymp}

In this section we state the main result on small-maturity forward-start option asymptotics.
First we need to define a number of functions. 
All functions below are real-valued and defined on $\mathbb{R}^*$. 
We recall that $\sgn(u)=1$ if $u\geq 0$ and -1 otherwise. 
\begin{equation}\label{eq:aB1hatsmalltime}
\left\{
\begin{array}{ll}
a_0(k) := \displaystyle \frac{\sgn(k)}{\sqrt{\beta_t}},
\quad 
a_1(k):=-\frac{a_0(k)\sqrt{v} \E^{-\kappa  t/2}}{2 \sqrt{|k|} \beta _t^{1/4}},
\quad 
a_2(k):=-\frac{\kappa\theta}{k \xi ^2}-\frac{\widehat{B}_1(a_0(k))}{a_0(k)}, & \\
a_3(k) := \displaystyle \frac{2\beta_t a_1^3(k)}{\xi^4 v^2}
\left[
\xi^2 v \beta _t \E^{\kappa  t} 
\left(|k|\xi^2 \beta_t^{\frac{1}{2}}\widehat{B}_1(a_0(k)) - k \xi^2 \widehat{B}_1'(a_0(k))-\kappa\theta\right)
+ (2\kappa\theta\beta_t\E^{\kappa t})^2 - \frac{\xi ^4 v^2}{16}
\right],& 
\end{array}
\right.
\end{equation}
where
\begin{equation}\label{eq:B1hat}
\widehat{B}_1(u) := \displaystyle \frac{u}{4}\Big(u^2\rho\xi-2\Big);
\end{equation}
\begin{equation}\label{eq:sigmarsmalltime}
\zeta(k):=\frac{2 \sqrt{v}  \E^{-\kappa  t/2}}{e_0(k)^{3/2}}, 
\qquad
r(k):= \frac{a_1^2(k)}{2}-\frac{\kappa\theta}{|k|\xi^2 \sqrt{\beta_t}},
\end{equation}
\begin{equation}\label{eq:esmalltime}
\left\{
\begin{array}{ll}
e_0(k) & := \displaystyle -2a_1(k)/a_0(k),\\
e_1(k) & := \displaystyle -2\beta_t r(k) , \\
e_2(k) & := \displaystyle -2\beta_t
\left(a_1(k)a_2(k)+a_0(k)a_3(k)+a_1(k)\widehat{B}_1'(a_0(k))\right),
\end{array}
\right.
\end{equation}
\begin{equation}\label{eq:psismalltime}
\left\{\begin{array}{rl}
\psi_0(k) & := \displaystyle \frac{a_0(k) v \E^{-\kappa t}}{e_0^3(k)}
\Big(e_0^2(k)+a_0(k)\beta _t \left[3 a_1(k)e_0(k)-2a_0(k) e_1(k)\right]\Big),
\\
\psi_1(k) & := \displaystyle  -4 a_0(k) v \beta _t \E^{-\kappa t}/e_0^4(k)\\
\psi_2(k) & := \displaystyle \frac{ v  \E^{-\kappa t}}{2 e_0^4(k)} 
\Big(4 a_0(k) \beta _t [3 a_0(k) e_1(k) - 4 a_1(k) e_0(k)]-5 e_0^2(k)\Big), \\
\psi_3(k) & := \displaystyle 8  v \beta _t \E^{-\kappa t}/e_0^5(k),\\
\psi_4(k) & := \displaystyle  \frac{v\E^{-\kappa t}}{2 e_0^3(k)}\left( \frac{e_1^2(k)-e_0(k) e_2(k)}{\beta_t}
- 2 a_0(k) a_1(k) e_0(k) e_1(k)+2 e_0^2(k) r(k)\right), \\
\end{array}
\right.
\end{equation}
\begin{equation}\label{eq:phi2smalltime}
\left\{
\begin{array}{rl}
\phi_2^{a}(k) & := \displaystyle\psi _2(k)-\frac{1}{2}\psi _0^2(k)
 - \frac{4\kappa\theta\beta_t}{\xi^2}\frac{2\kappa\theta+\xi^2}{e_0^2(k) \xi^2}
 - \frac{4\kappa\theta\beta_t}{\xi^2}\frac{a_0(k)\psi_0(k)}{e_0(k)},  \\
\phi_2^{b}(k) & := \displaystyle \psi _3(k) - \psi _0(k) \psi _1(k)
 - \frac{4\kappa\theta\beta_t}{\xi^2}\frac{a_0(k)\psi_1(k)}{e_0(k)}, \\
\phi_2^{c}(k) & := -\psi _1^2(k)/2, 
\end{array}
\right.
\end{equation}
\begin{equation}\label{eq:zpsmalltime}
z_1(k):=\psi_4(k)-a_3(k) k - \frac{2\kappa\theta}{\xi^2}\frac{e_1(k)}{e_0(k)},\quad
p_1(k):=e_0(k) + \frac{\phi_2^{a}(k)}{\zeta^2(k)}+\frac{3\phi_2^{b}(k)}{\zeta^4(k)}+\frac{15\phi_2^{c}(k)}{\zeta^6(k)},
\end{equation}
\begin{equation}\label{eq:csmalltime}
\left\{
\begin{array}{rll}
c_0(k) & := \displaystyle 2 |a_1(k) k|,
& c_1(k) := \displaystyle \frac{v\E^{-\kappa t}}{e_0(k)}\left(a_0(k)a_1(k) - \frac{e_1(k)}{2\beta_t e_0(k)}\right)-a_2(k) k , \\ 
c_2(k) & := \displaystyle e_0(k)^{-2\kappa\theta/\xi^2},
& c_3(k) := z_1(k)+p_1(k).
\end{array}
\right.
\end{equation}
We now state the main result of the section, i.e. an asymptotic expansion formula for forward-start
option prices as the remaining maturity tends to zero.
The proof is given in Section~\ref{sec:optpricefwdsmile}.
\begin{theorem}\label{theorem:fwdstartoptionssmalltime}
The following expansion holds for forward-start option prices for all $k\in\mathbb{R}^*$ as $\tau$ tends to zero:
\begin{align*}
&\mathbb{E}\left[\left(\E^{X^{( t)}_{\tau}}-\E^k\right)^+\right]
 = \left(1-\E^{k}\right)\ind_{\{k<0\}} \\
& + \exp{\left({-\frac{\Lambda^*(k)}{\sqrt{\tau}}+\frac{c_0(k)}{\tau^{1/4}}+c_1(k)}+k\right)}
\frac{\beta_t \tau^{\left(7/8-\theta\kappa/(2\xi^2)\right)}c_2(k)}{\zeta(k) \sqrt{2\pi}}\left(1+c_3(k)\tau^{1/4}+o\left(\tau^{1/4}\right)\right),
\end{align*}
where $\Lambda^*$ is characterised in Lemma~\ref{lemma:U*Characterisation}, $c_0,\ldots,c_3$ in~\eqref{eq:csmalltime}, $\zeta$ in~\eqref{eq:sigmarsmalltime} and $\beta_t$ in~\eqref{eq:DGammaBeta}.
\end{theorem}

\begin{remark} \
\begin{enumerate}[(i)]
\item
We have $\Lambda^*(k)>0$ and $c_0(k)>0$ for all $k\in\mathbb{R}^*$. 
Also note that $\Lambda^*$ is linear as opposed to being strictly convex in the BSM model, see Lemma~\ref{lemma:BSAsympSmallTau} below.
\item
The forward time-scale $\sqrt{\tau}$ results in out-of-the-money forward-start options decaying as $\tau$ tends to zero at leading order with a rate of $\exp{(-1/\sqrt{\tau})}$ as opposed to a rate of $\exp{(-1/\tau)}$ in the BSM model. 
\item
The fact that the limiting forward lmgf is non-steep (trivially zero on a bounded interval) results in a different asymptotic regime for higher order terms compared to the BSM model. 
In particular we have a $\tau^{1/4}$ dependence as opposed to a $\tau$ dependence in the BSM model and the introduction of the parameter dependent term $\tau^{(7/8-\theta\kappa/(2\xi^2))}$. 
The implications of this parameter dependent term for forward-smile asymptotics will be discussed further in Remark~\ref{rem:smmatfwdsmileasymp}(vii).
\item
The asymptotic expansion is given in closed-form and can in principle be extended to arbitrary order using the methods given in the proof.
\end{enumerate}
\end{remark}

As an immediate consequence of Theorem~\ref{theorem:fwdstartoptionssmalltime} 
we have the following corollary, which provides an example of a family of random variables for which 
the limiting logarithmic moment generating function is zero (on its effective domain) but a large deviation principle still holds.
This is to be compared to the G\"artner-Ellis theorem~\cite[Theorem 2.3.6]{DZ93} which requires the limiting
lmgf to be at least steep at the boundaries of its effective domain for an LDP to hold.
\begin{corollary}
$\left(X_{\tau}^{(t)}\right)_{\tau\geq0}$ satisfies an LDP with speed $\sqrt{\tau}$ and good rate function $\Lambda^*$ as $\tau$ tends to zero.
\end{corollary}

\begin{proof}

The proof of Theorem~\ref{theorem:fwdstartoptionssmalltime} holds with only minor modifications for digital options,
which are equivalent to probabilities of the form 
$\mathbb{P}\left(X_{\tau}^{(t)}\leq k\right)$ or $\mathbb{P}\left(X_{\tau}^{(t)}\geq k\right)$.
One can then show that 
$\lim_{\tau\searrow 0}\sqrt{\tau}\log \mathbb{P}\left(X_{\tau}^{(t)}\leq k\right)
 = -\inf\{\Lambda^*(x), x\leq k\}$.
Note that of course this infimum is null whenever $k>0$.
Consider now an open interval of the real line of the form $(a,b)$.
Since $(a,b) = (-\infty, b)\setminus (-\infty,a]$, then by continuity of the function $\Lambda^*$
and its properties given in Lemma~\ref{lemma:U*Characterisation}, we immediately obtain that
$$
\lim_{\tau\searrow 0}\sqrt{\tau}\log \mathbb{P}\left(X_{\tau}^{(t)} \in (a,b)\right)
 = -\inf_{x\in (a,b)}\Lambda^*(x).$$
Since any Borel set of the real line can be written as a (countable) union / intersection of open intervals, 
the corollary follows from the definition of the large deviations principle, see~\cite[Section 1.2]{DZ93}.
\end{proof}

In order to translate the forward-start option results into forward smile asymptotics we require a similar expansion for the BSM model. 
The following lemma is a direct consequence of~\cite[Corollary 2.9]{JR12} and the proof is therefore omitted.

\begin{lemma}\label{lemma:BSAsympSmallTau}
In the BSM model the following expansion holds for all $k\in\mathbb{R}^*$ as $\tau$ tends to zero:
$$
\mathbb{E}\left[\left(\E^{X^{( t)}_{\tau}}-\E^k\right)^+\right] 
=
\left(1-\E^{k}\right)\ind_{\{k<0\}}+
\frac{\E^{k/2-k^2/(2\sigma^2 \tau)}\left(\sigma^2\tau\right)^{3/2}}{k^2\sqrt{2\pi}}\left[1-\left(\frac{3}{k^2}+\frac{1}{8}\right)\sigma^2\tau+o(\tau)\right].
$$
\end{lemma}

\section{Small-maturity forward smile asymptotics} \label{sec:smfwdsmileasymp}

\subsection{Out-of-the-money forward implied volatility}\label{sec:otmasymp}

We now translate the small-maturity forward-start option asymptotics into forward smile asymptotics. 
Define the functions $v_i:\mathbb{R}^*\times\mathbb{R}_{+}^*\to\mathbb{R}$ ($i=0,1,2,3$) by
\begin{equation}\label{eq:vov1v2smalltimehest}
\left.
\begin{array}{rll}
v_0(k,t) & := \displaystyle \frac{k^2}{2 \Lambda^*(k)}=\frac{\sqrt{\beta_t}|k|}{2} ,\quad 
v_1(k,t):=\frac{2 c_0(k) v_0^2(k,t)}{k^2}=\frac{\E^{-\kappa t/2}\beta_t^{1/4}\sqrt{v|k|}}{2}, \\
v_2(k,t) & := \displaystyle \frac{2 v_0^2(k,t)}{k^2}\log \left(\frac{\E^{c_1(k)} c_2(k) \beta_t k^2}{\zeta(k)  v_0^{3/2}(k,t)}\right)+\frac{v_0^2(k,t)}{k}+\frac{v_1^2(k,t)}{v_0(k,t)},\\
v_3(k,t) & := \displaystyle \frac{v_0(k,t)}{k^2}\Big(2 c_3(k) v_0(k,t)-3 v_1(k,t)\Big)+\frac{v_1(k,t)}{v_0(k,t)}\left(2 v_2(k,t)-\frac{v_1^2(k,t)}{v_0(k,t)}\right),
\end{array}
\right.
\end{equation}
with $\Lambda^*$ characterised in Lemma~\ref{lemma:U*Characterisation}, $c_0,\ldots,c_3$ in~\eqref{eq:csmalltime}, $\zeta$ in~\eqref{eq:sigmarsmalltime} and $\beta_t$ in~\eqref{eq:DGammaBeta}. 
On~$\mathbb{R}^*$, $\Lambda^*(k)>0$ and so $v_0(k,t)>0$. 
Further $c_0(k)>0$ and so $v_1(k,t)>0$. 
Also $c_2(k)>0$ and $\zeta(k)>0$ so that $v_2$ is a well defined real-valued function. 
The following theorem---proved in Section~\ref{sec:optpricefwdsmile}---is the main result of the section.
\begin{theorem}\label{theorem:fwdsmilesmalltime}
The following expansion holds for the forward smile for all $k\in\mathbb{R}^*$ as $\tau$ tends to zero:
\begin{equation*}
\sigma_{t,\tau}^2(k)=
\left\{
\begin{array}{ll}
\displaystyle \frac{v_0(k,t)}{\tau^{1/2}}+\frac{v_{1}(k,t)}{\tau^{1/4}}
+o\left(\frac{1}{\tau^{1/4}}\right), & \text{if } 4\kappa\theta\neq\xi^2,\\
\displaystyle \frac{v_0(k,t)}{\tau^{1/2}}+\frac{v_{1}(k,t)}{\tau^{1/4}}
+v_{2}(k,t)+v_{3}(k,t)\tau^{1/4}+o\left(\tau^{1/4}\right), & \text{if } 4\kappa\theta = \xi^2.
\end{array}
\right.
\end{equation*}
\end{theorem}

\newpage
\begin{remark} \ \label{rem:smmatfwdsmileasymp}
\begin{enumerate}[(i)]
\item
Note that $v_0(k,t)$ and $v_1(k,t)$ are strictly positive for all $k\in\mathbb{R}^*$, 
so that the Heston forward smile blows up to infinity (except ATM) as $\tau$ tends to zero.
\item
Both $v_0(\cdot,t)$ and $v_1(\cdot,t)$ are even functions and correlation-independent quantities so that for small maturities the Heston forward smile becomes symmetric (in log-strikes) around the ATM point.
Consequently, if one believes that the small-maturity forward smile should be downward sloping (similar to the spot smile) then the Heston model should not be chosen.
This small-maturity 'U-shaped' effect for the Heston forward smile has been mentioned qualitatively by practitioners; see~\cite{B02}.
\item
We use the notation $f\sim g$ to mean $f/g =1$ as $\tau\to0$. 
Then in Heston  we have $\sigma_{t,\tau}^2 \sim \sqrt{\beta_t}|k|/(2\sqrt{\tau}) $ and 
in exponential L\'evy models with L\'evy measure $\nu$ satisfying $\text{supp}\, \nu=\mathbb{R}$ we have  $\sigma_{t,\tau}^2 \sim -k^2/(2\tau \log \tau) $ ~\cite[Page 21]{Tank}.
We therefore see that the small-maturity exponential L\'evy smile blows up at a much quicker rate than the Heston forward smile.
\item
We have $\lim_{k\to0}v_0(k,t)=v_0(0,t)=0$ and $\lim_{k\to0}v_1(k,t)=v_1(0,t)=0$.
Higher-order terms are not necessarily continuous at $k=0$. 
For example (when $4\kappa\theta=\xi^2$) we have $\lim_{k\to0}v_2(k,t)=+\infty$.
\item
The at-the-money forward implied volatility ($k=0$) asymptotic is not covered by Theorem~\ref{theorem:fwdsmilesmalltime} and a separate analysis is needed for this case (see Section~\ref{sec:ATMfwdsmile}). 
In particular the proof fails since in this case the key function $u^*_{\tau}(0)$ 
(defined through equation~\eqref{eq:u^*tausmalltime}) does not converge to a boundary point, 
but rather to zero as $\tau$ tends to zero (see the proof of Lemma~\ref{lemma:u^*tausmalltime}).
\item
It does not make sense to consider the limit of our asymptotic result for fixed $k\in\mathbb{R}^*$ as $t$ tends to zero since for $t=0$ using the forward time-scale $h(\tau)\equiv\sqrt{\tau}$ will produce a trivial limiting lmgf and hence none of the results will carry over.
The time scale in the spot case is $h(\tau)\equiv\tau$; see~\cite{FJSmall}.
Our result is only valid in the forward (not spot) smile case.
\item
As seen in the proof, due to the term $\tau^{7/8-\theta\kappa/(2\xi^2)}$ in the forward-start option asymptotics in Theorem~\ref{theorem:fwdstartoptionssmalltime}, one can only specify the small-maturity forward smile to arbitrary order if $4\kappa\theta=\xi^2$. 
If this is not the case then such an expansion for the forward smile only holds up to order~$\mathcal{O}(1/\tau^{1/4})$. 
Now the dynamics of the Heston volatility $\sigma_t:=\sqrt{V_t}$ is given by
$
\D \sigma_t=\left(\frac{4\kappa\theta-\xi^2}{8\sigma_t}-\frac{\kappa \sigma_t}{2}\right)\D t +\frac{\xi}{2} \D W_t,
$
with $\sigma_0=\sqrt{v}$.
If we set $4\kappa\theta=\xi^2$ the volatility becomes Gaussian, which corresponds to a specific case of the Sch\"obel-Zhu stochastic volatility model~\cite{Zhu}.
So as the Heston volatility dynamics deviate from Gaussian volatility dynamics a certain degeneracy occurs such that one cannot specify high order forward smile asymptotics in the small-maturity case. 
Interestingly, a similar degenerary occurs when studying the tail probability of the stock price.
As proved in~\cite{DFJV}, the square-root behaviour of the variance process induces some singularity and hence
a fundamentally different behaviour when $4\kappa\theta\ne \xi^2$.
\end{enumerate}
\end{remark}

\subsection{At-the-money forward implied volatility}\label{sec:ATMfwdsmile}
The analysis above excluded the at-the-money case $k=0$.
We show below that this case has a very different behaviour and can be studied with a much simpler machinery.
In this section, we shall denote the future implied volatility $\sigma_t(k,\tau)$ as 
the implied volatility corresponding to a European call/put option
with strike $\E^{k}$, maturity $\tau$, observed at time $t$.
We first start with the following model-independent lemma, bridging the gap between the at-the-money 
future implied volatility~$\sigma_t(0,\tau)$
and the forward implied volatility~$\sigma_{t,\tau}(0)$.
Note that a similar result---albeit less general---was derived in~\cite{LC09}.
We shall denote by~$\mathbb{E}_0$ the expectation (under the given risk-neutral probability measure) 
with respect to~$\mathcal{F}_0$, the filtration at time zero.
\begin{lemma}\label{lem:modelindatm}
Let $t>0$.
Assume that there exists $n\in\mathbb{N}^*$ such that the expansion
$\sigma_t(0,\tau) = \sum_{j=0}^{n}\sigma_j(t)\tau^{j} + o\left(\tau^{n}\right)$
holds and that $\mathbb{E}_0\left(\sigma_j(t)\right)<\infty$
for $j=0,...,n$.
If the at-the-money forward implied volatility satisfies
$\sigma_{t,\tau}(0) = \sum_{j=0}^{n}\bar{\sigma}_j(t)\tau^{j} + o\left(\tau^{n}\right)$, then 
$\bar{\sigma}_j(t) = \mathbb{E}_0 (\sigma_j(t))$ for all $j=0,\ldots,n$.
\end{lemma}

\begin{proof}
In the Black-Scholes model, we know that for any $t\geq 0$, $\tau>0$, 
the price at time~$t$ of a (re-normalised) European call option with maturity $t+\tau$ is 
$
C_t^{\textrm{BS}}(k,\tau,\sigma) := \mathbb{E}\left[\left(S^{\mathrm{BS}}_{t+\tau}/S^{\mathrm{BS}}_t-\E^{k}\right)_+\vert \Ff_t\right]$,
and its at-the-money expansion as the maturity $\tau$ tends to zero reads (see~\cite[Corollary 3.5]{FJL11})
$$
C_t^{\textrm{BS}}(0,\tau,\sigma) = \frac{1}{\sqrt{2\pi}}\left(\sigma\sqrt{\tau}-\frac{\sigma^3 \tau^{3/2}}{24}+\mathcal{O}\left(\sigma^5\tau^{5/2}\right)\right).
$$
We keep the $\sigma$ dependence in the $\mathcal{O}(\ldots)$ to highlight the fact that, when $\sigma$ depends on $\tau$ 
(such as $\sigma=\sigma_t(0,\tau)$), one has to be careful not to omit some terms.
Now, for a given martingale model for the stock price~$S$, we shall denote by $C_{t}(k,\tau)$ the price at time $t$ of a European call option with payoff 
$\left(S_{t+\tau}/S_t-\E^{k}\right)_+$
at time $t+\tau$.
The future implied volatility $\sigma_t(k,\tau)$ is then the unique solution to 
$C_t^{\textrm{BS}}(k,\tau,\sigma_t(k,\tau)) = C_t(k,\tau)$.
For at-the-money $k=0$, we obtain the following expansion for short maturity~$\tau$:
\begin{equation}\label{eq:AppExp}
C_t(0,\tau) = 
\frac{1}{\sqrt{2\pi}}\left(
\sigma_0(t)\sqrt{\tau}+\left(\sigma_1(t)-\frac{\sigma_0^3(t)}{24}\right)\tau^{3/2}+
\mathcal{O}\left(\tau^{5/2}\right)
\right),
\end{equation}
where we have used here the expansion assumed for $\sigma_t(0,\tau)$.
Note also that the coefficients $\sigma_j(t)$ are random variables.
We follow the probabilistic version of the $\mathcal{O}$ notation detailed in~\cite[Section 5]{Janson}, 
namely the random remainder~$R_\tau$ is $\mathcal{O}_P(\tau^{5/2})$ as $\tau$ tends to zero if and only if
for any $\varepsilon>0$ there exist a constant $c_\varepsilon>0$ and a threshold $\tau_\varepsilon>0$ for which
$\mathbb{P}\left(|R_\tau|\leq c_\varepsilon \tau^{5/2}\right)>1-\varepsilon$ for all $\tau<\tau_\varepsilon$.
For brevity we abuse the notations slightly here and write $\mathcal{O}$ instead of $\mathcal{O}_P$.
Now, the forward-start European call option (at inception) in the Black-Scholes model reads
$$
 \mathbb{E}_0\left[\left(\frac{S^{\mathrm{BS}}_{t+\tau}}{S^{\mathrm{BS}}_t}-\E^{k}\right)_+\right]
 = \mathbb{E}_0\left\{\mathbb{E}\left[\left(\frac{S^{\mathrm{BS}}_{t+\tau}}{S^{\mathrm{BS}}_t}-\E^{k}\right)_+\vert \Ff_t\right]\right\}
 = \Nn(d_+(\sigma, \tau))-\E^{k}\Nn(d_-(\sigma,\tau))
=C_0^{\mathrm{BS}}(k, \tau, \sigma),
$$
where $d_{\pm}(\sigma,\tau):=(-k\pm\sigma^2\tau/2)/(\sigma\sqrt{\tau})$.
For a given model, let us denote by $\overline{C}(k, t, \tau)$ the price of a forward-start European call option.
By definition of the forward implied volatility $\sigma_{t,\tau}(k)$, we have
$\overline{C}(k,t,\tau) = C_0^{\mathrm{BS}}(k,\tau,\sigma_{t,\tau}(k))$.
For at-the-money $k=0$, it follows that
$$
\overline{C}(0,t,\tau) = C_0^{\mathrm{BS}}(0, \tau, \sigma_{t,\tau}(0)) = \Nn(d_+(\sigma_{t,\tau}(0), \tau))-\Nn(d_-(\sigma_{t,\tau}(0),\tau)).
$$
Using the assumed expansion for $\sigma_{t,\tau}(0)$ we similarly obtain (as in~\eqref{eq:AppExp})
\begin{equation}\label{eq:AppExp2}
\overline{C}(0,t,\tau) = 
\frac{1}{\sqrt{2\pi}}\left(
\bar\sigma_0(t)\sqrt{\tau}+\left(\bar\sigma_1(t)-\frac{\bar\sigma_0^3(t)}{24}\right)\tau^{3/2}+\mathcal{O}\left(\tau^{5/2}\right)
\right).
\end{equation}
Note that now the coefficients $\bar\sigma_j(t)$ are not random variables, but simple constants.
Recall now that 
\begin{equation}\label{eq:FSeq}
\overline{C}(k, t, \tau)
 := \mathbb{E}_0\left[\left(\frac{S_{t+\tau}}{S_t}-\E^{k}\right)_+\right]
 = \mathbb{E}_0\left\{\mathbb{E}\left[\left(\frac{S_{t+\tau}}{S_t}-\E^{k}\right)_+\vert \Ff_t\right]\right\}
 = \mathbb{E}_0\left(C_{t}(k,\tau)\right).
\end{equation}
Combining this with~\eqref{eq:AppExp} and~\eqref{eq:AppExp2}, we find that $\bar{\sigma}_j(t)=\mathbb{E}_0\left(\sigma_j(t)\right)$ for $j=0,1$. 
The higher-order terms for the expansion can be proved analogously and the lemma follows.
\end{proof}

We now apply this to the Heston model.
Recall the definition of the Kummer (confluent hypergeometric) function $\mathrm{M}:\mathbb{C}^3\to\RR$:
$$\mathrm{M}\left(\alpha, \mu,z\right):=\sum_{n\geq 0}\frac{(\alpha)_n}{(\mu)_n}\frac{z^n}{n!},\quad \mu\neq 0,-1,...,$$
where the Pochhammer symbol is defined by $(\alpha)_n:=\alpha\left(\alpha+1\right)\cdots\left(\alpha+n-1\right)$
for $n\geq 1$ and $(\alpha)_0=1$. 
For any $p>-2\kappa\theta/\xi^2$ and $t>0$ we define
\begin{equation}\label{eq:HestonMoments}
\Delta(t,p):=2^p\beta_t^p \exp\left({-\frac{v\E^{-\kappa t}}{2\beta_t}}\right)\frac{\Gamma\left(2\kappa\theta/\xi^2+p\right)}{\Gamma\left(2\kappa\theta/\xi^2\right)}\mathrm{M}\left(\frac{2\kappa\theta}{\xi^2}+p,\frac{2\kappa\theta}{\xi^2},\frac{v\E^{-\kappa t}}{2\beta_t}\right),
\end{equation}
with $\beta_t$ defined in~\eqref{eq:DGammaBeta}. 
This function is related to the moments of the Feller diffusion (see~\cite[Theorem 2.4]{D01}):
for any $t>0$, $\mathbb{E}\left[V_t^p\right]=\Delta(t,p)$ if $p>-2\kappa\theta/\xi^2$ and is infinite otherwise.
Note in particular that $\lim_{t\searrow 0}\Delta(t,p) = v^{p}$ (see~\cite[13.1.4 page 504]{Abra}).
The Heston forward at-the-money volatility asymptotic is given in the following theorem.

\begin{theorem}\label{theorem:fwdatmsmalltime}
The following expansion holds for the forward at-the-money volatility as $\tau$ tends to zero:
\begin{equation*}
\sigma_{t,\tau}(0)=
\left\{
\begin{array}{ll}
\displaystyle \Delta(t,1/2)
+o\left(1\right), & \text{if } 4\kappa\theta\leq\xi^2,\\
\displaystyle \Delta(t,1/2)
+\frac{\Delta(t,-1/2)}{4}\left(\kappa\theta+\frac{\xi^2}{24}(\rho^2-4)\right)\tau
+\frac{\Delta(t,1/2)}{8} (\rho\xi-2\kappa )\tau
+o(\tau), & \text{if } 4\kappa\theta > \xi^2.
\end{array}
\right.
\end{equation*}
\end{theorem}

\begin{remark} \
\begin{enumerate}[(i)]
\item As opposed to the out-of-the-money case, the small-maturity limit here is well defined.
\item Combining Lemma~\ref{lem:modelindatm} and~\cite{BBF},
$\lim_{ \tau \searrow 0}\sigma_{t,\tau}(0)=\mathbb{E}(\sqrt{V_t})$ holds for any well-behaved stochastic volatility model $(S,V)$.
\item The proof does not allow one to conclude any information about higher order terms in Heston for the case $4\kappa\theta\leq\xi^2$.
A different method would need to be used to compute higher order asymptotics in this case.
\end{enumerate}
\end{remark}

\begin{proof}[Proof of Theorem~\ref{theorem:fwdatmsmalltime}]
In Heston we recall (\cite[Corollary 4.3]{FJL11} or~\cite[Corollary 3.3]{JR12}) the asymptotic 
$
\sigma_t^2(0,\tau)=V_t+
\left(\frac{\kappa\theta}{2}+\frac{\xi^2}{48}\left(\rho^2-4\right)+\frac{V_t}{4}(\rho\xi-2\kappa)\right)\tau
+o(\tau),
$
and so for small $\tau$ we have
$\sigma_t(0,\tau)=\sigma_0(t)+\sigma_1(t)\tau+o(\tau)$,
with $\sigma_0(t):=\sqrt{V_t}$ and $\sigma_1(t):=\frac{1}{4\sqrt{V_t}}\left(\kappa\theta+\frac{\xi^2}{24}(\rho^2-4)\right)
+\frac{\sqrt{V_t}}{8} (\rho\xi-2\kappa)$.
Lemma~\ref{lem:modelindatm} and~\eqref{eq:HestonMoments} conclude the proof.
\end{proof}

\section{Numerics}\label{sec:smhnumerics}

We first compare the true Heston forward smile and the asymptotics developed in the paper.
We calculate forward-start option prices using the inverse Fourier transform representation 
in~\cite[Theorem 5.1]{L04F}
and a global adaptive Gauss-Kronrod quadrature scheme. 
We then compute the forward smile $\sigma_{t,\tau}$ with a simple root-finding algorithm.
The Heston model parameters are given by $\rho=-0.8$, $\xi=0.52$, $\kappa=1$ and $v=\theta=0.07$ unless otherwise stated in the figures. 
In Figures~\ref{fig:halfmonth} and~\ref{fig:onemonth} we compare the true forward smile using Fourier inversion and the asymptotic in Theorem~\ref{theorem:fwdsmilesmalltime}.
It is clear that the small-maturity asymptotic has very different features relative to "smoother" asymptotics derived in~\cite{JR12}.
This is due to the introduction of the forward time-scale and to the fact that the limiting lmgf is not steep.
Note also from Remark~\ref{rem:smmatfwdsmileasymp}(iv) that the asymptotics in Theorem~\ref{theorem:fwdsmilesmalltime} can approach zero or infinity as the strike approaches at-the-money.
This appears to be a fundamental feature of non-steep asymptotics; 
numerically this implies that the asymptotic may break down for strikes in a region around the at-the-money point.
In Figure~\ref{fig:SmallMatATMFourier} we compare the true at-the-money forward volatility using Fourier inversion and the asymptotic in Lemma~\ref{lem:modelindatm}. 
Results are in line with expectations and the at-the-money asymptotic is more accurate than the out-of-the-money asymptotic. 
This is because the at-the-money forward volatility (unlike the out-of-the-money case) has a well defined limit as $\tau$ tends to zero.
In Figure~\ref{fig:exploding} we use these results to gain intuition on how the Heston forward smile explodes for small maturities. 
In~\cite{JR12} the authors derived a diagonal small-maturity asymptotic expansion for the Heston forward smile valid for small forward start-dates and small maturities. 
In order for the small-maturity asymptotic in this paper to be useful, there needs to be a sufficient amount of variance of variance at the forward-start date.
Practically this means that the asymptotic performs better as one increases the forward-start date.
On the other hand the diagonal-small maturity asymptotic expansion is valid for small forward-start dates.
In this sense these asymptotics complement each other.
Figure~\ref{fig:smallmatvsdiag} shows the consistency of these two results for small forward-start date and maturity.

\begin{figure}[h!tb] 
\centering
\mbox{\subfigure[Asymptotic vs Fourier inversion.]{\includegraphics[scale=0.8]{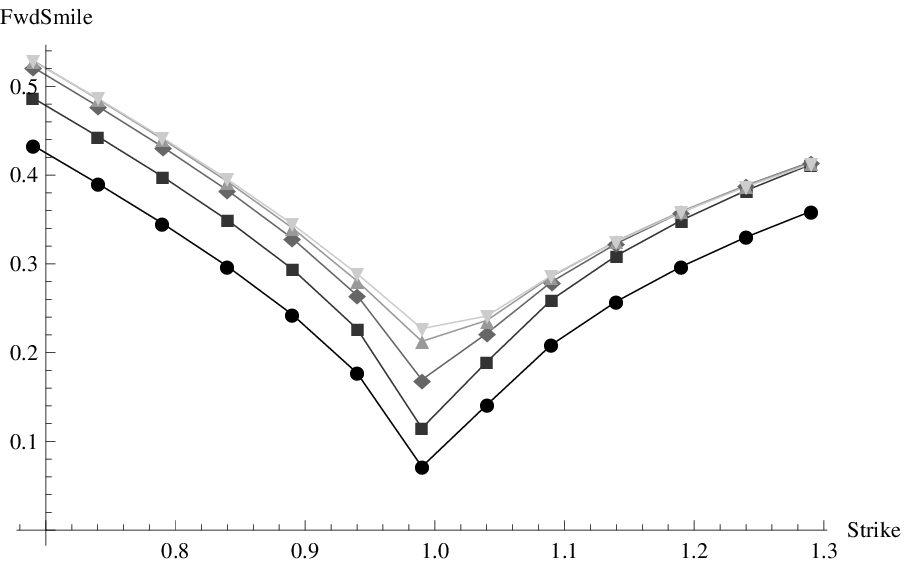}}\quad
\subfigure[Errors.]{\includegraphics[scale=0.8]{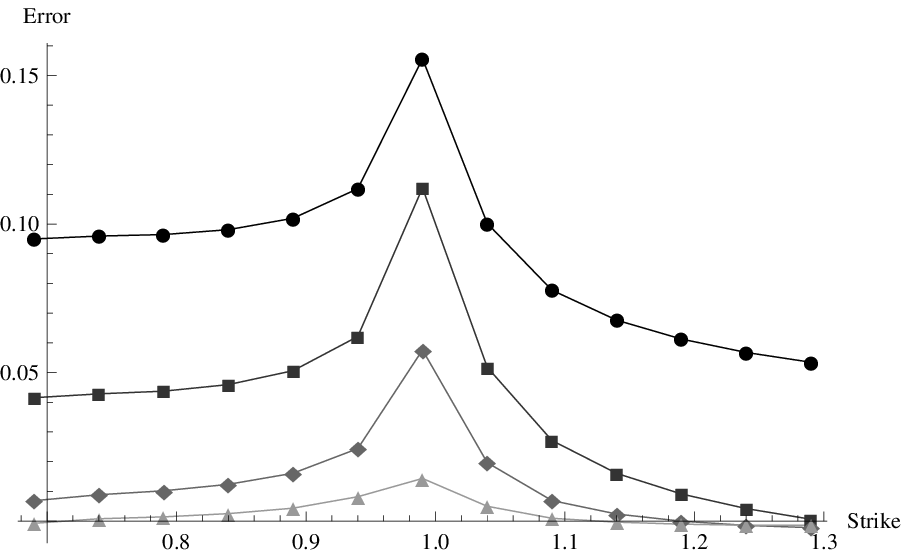}}}
\caption{Here $t=1$ and $\tau=1/24$.
In (a) circles, squares, diamonds and triangles represent the zeroth, first, second and third-order asymptotics respectively and backwards triangles represent the true forward smile using Fourier inversion. 
In (b) we plot the errors.}
\label{fig:halfmonth}
\end{figure}

\begin{figure}[h!tb] 
\centering
\mbox{\subfigure[Asymptotic vs Fourier inversion.]{\includegraphics[scale=0.8]{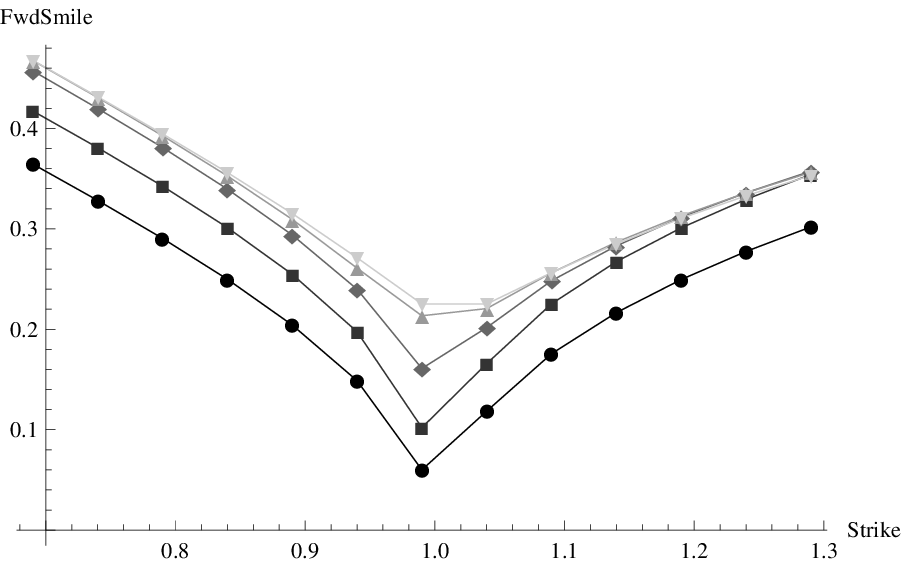}}\quad
\subfigure[Errors.]{\includegraphics[scale=0.8]{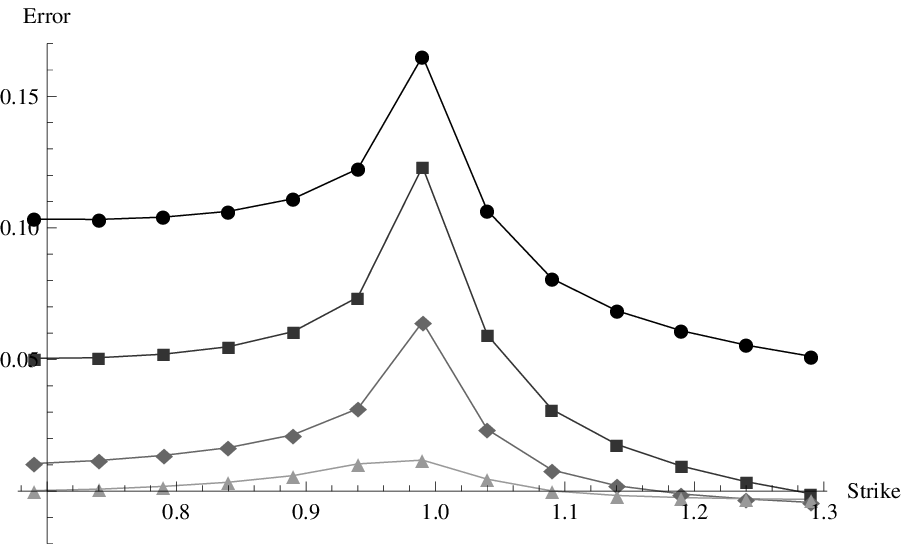}}}
\caption{Here $t=1$ and $\tau=1/12$. 
In (a) circles, squares, diamonds and triangles represent the zeroth, first, second and third-order asymptotics respectively and backwards triangles represent the true forward smile using Fourier inversion. 
In (b) we plot the errors. }
\label{fig:onemonth}
\end{figure}

\newpage

\begin{figure}[h!tb] 
\centering
\mbox{\subfigure[ATM Asymptotic vs Fourier inversion.]{\includegraphics[scale=0.7]{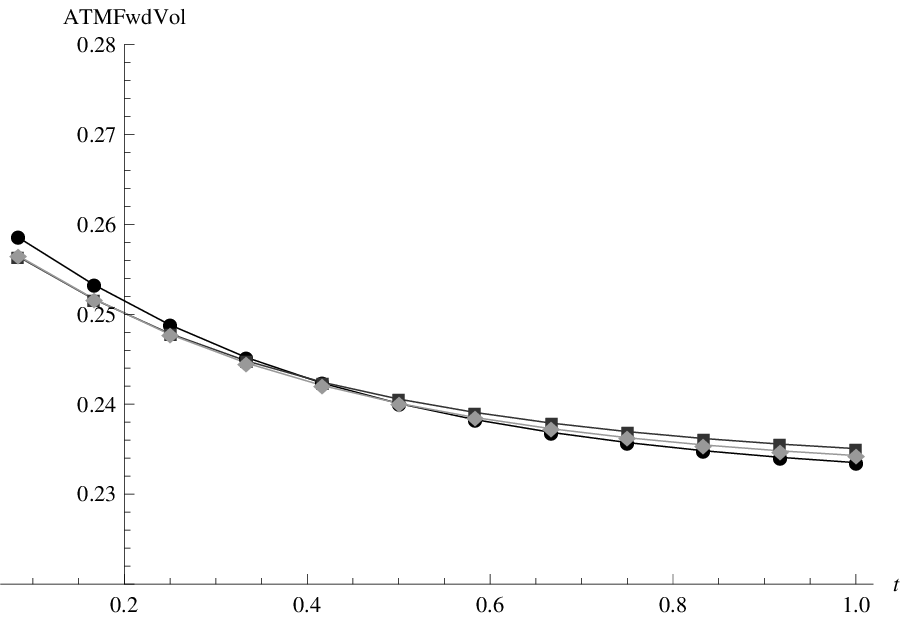}}\quad
\subfigure[Errors.]{\includegraphics[scale=0.7]{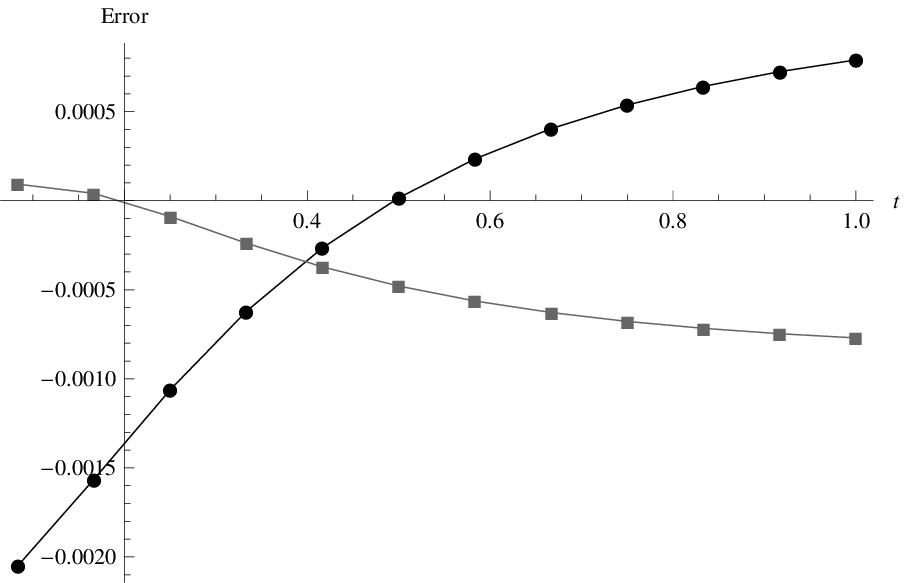}}}
\caption{Plot of the forward ATM volatility ($\tau=1/12$) as a function of the forward-start date $t$. 
The Heston parameters are $\rho=-0.6$, $\kappa=1$, $\xi=0.4$ and $v=\theta=0.07$. 
In (a) circles, squares and diamonds are the zeroth-order, 
the first-order and the true forward ATM volatility.}
\label{fig:SmallMatATMFourier}
\end{figure}

\begin{figure}[h!tb] 
\centering
\mbox{\subfigure[$t=1$, $\tau=1/100$.]{\includegraphics[scale=0.7]{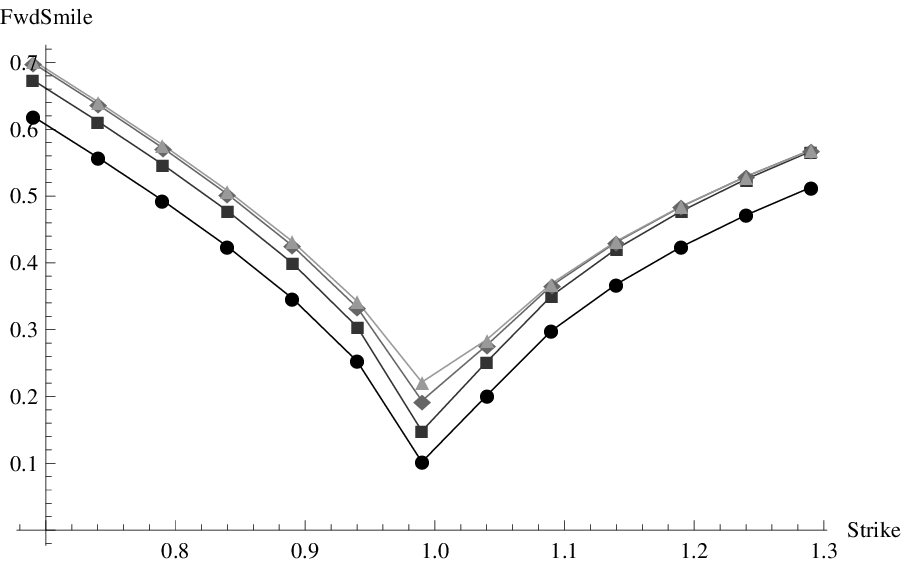}}\quad
\subfigure[$t=1$, $\tau=1/1000$.]{\includegraphics[scale=0.7]{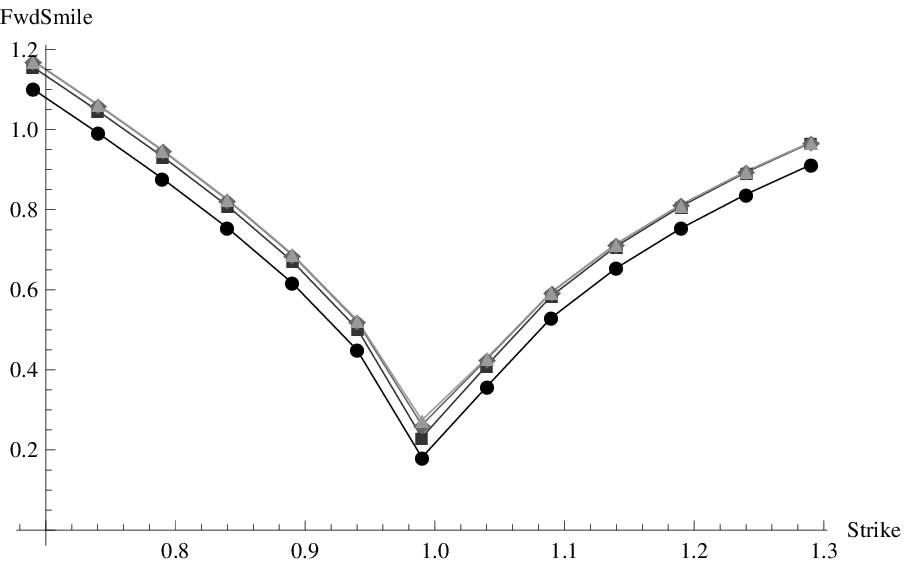}}}
\caption{Circles, squares, diamonds and triangles represent the zeroth, first, second and third-order asymptotics respectively.}
\label{fig:exploding}
\end{figure}

\begin{figure}[h!tb] 
\centering
\includegraphics[scale=0.7]{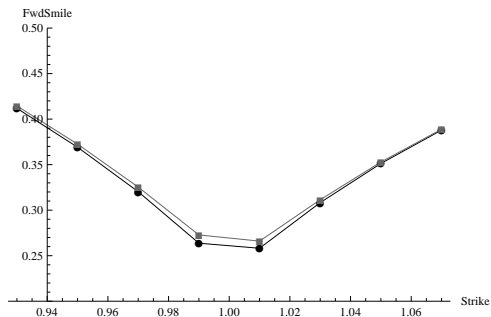}
\caption{Here we compare the small-maturity third-order asymptotic (circles) to the diagonal small-maturity 
second-order asymptotic of~\cite{JR12} (squares) for $t=1/12$ and $\tau=1/1000$.}
\label{fig:smallmatvsdiag}
\end{figure}

\section{Proof of Theorems~\ref{theorem:fwdstartoptionssmalltime} and~\ref{theorem:fwdsmilesmalltime}}\label{sec:proofssmalltau}

We split the proof of Theorems~\ref{theorem:fwdstartoptionssmalltime} and~\ref{theorem:fwdsmilesmalltime} into several parts, from Section~\ref{sec:hestfwdtime} to Section~\ref{sec:optpricefwdsmile} below. 
In Section~\ref{sec:hestfwdtime} we develop the necessary tools to characterise the small-maturity Heston forward lmgf domain and derive the Heston forward time-scale (Lemma~\ref{Lemma:HestFwdTimeScale}). 
In Section~\ref{sec:asympmeaschange} we use the forward time-scale to define a time-dependent asymptotic measure-change and derive expansions for fundamental auxiliary functions needed in the analysis. 
In Section~\ref{sec: Charactfuncasymp} we derive the asymptotics of the characteristic function of a re-scaled version of the forward price process $(X_{\tau}^{(t)})$ under the asymptotic measure-change defined in Section~\ref{sec:asympmeaschange}. 
This section also uses Fourier transform methods to derive asymptotics of important expectations using this characteristic function expansion. 
Section~\ref{sec:optpricefwdsmile} finally puts all the pieces together and proves 
Theorems~\ref{theorem:fwdstartoptionssmalltime} and~\ref{theorem:fwdsmilesmalltime}.

\subsection{Heston forward time-scale}\label{sec:hestfwdtime}

A straightforward application of the tower property for expectations (see also~\cite{Hong}) yields
the Heston forward logarithmic moment generating function:
\begin{equation}\label{eq:LambdaTau}
\Lambda_{\tau}^{(t)}(u,1)
:=\log\mathbb{E}\left(\E^{u X_{\tau}^{(t)}}\right)
=
A\left(u,\tau\right)+\frac{B(u,\tau)}{1-2\beta_t B(u,\tau)}v\E^{-\kappa t}
-\frac{2\kappa\theta}{\xi^2}\log\left(1-2\beta_t B(u,\tau)\right),
\end{equation}
for all $u\in\mathcal{D}_{t,\tau}$, where 
\begin{equation}\label{eq:ABFunctions}
\left.\begin{array}{ll}
A\left(u,\tau\right) & := \displaystyle
\frac{\kappa\theta}{\xi^2}\left(\left(\kappa-\rho\xi u- d(u)\right)\tau-2\log\left(\frac{1-\gamma(u)\exp\left(-d(u)\tau\right)}{1-\gamma(u)}\right)\right),\\
B(u,\tau) & := \displaystyle
\frac{\kappa-\rho\xi u-d(u)}{\xi^2}\frac{1-\exp\left(-d(u)\tau\right)}{1-\gamma(u)\exp\left(-d(u)\tau\right)},
\end{array}
\right.
\end{equation}
\begin{equation}\label{eq:DGammaBeta}
d(u) := \left(\left(\kappa-\rho\xi u\right)^2+u\left(1-u\right)\xi^2\right)^{1/2},
\quad
\gamma(u) := \frac{\kappa-\rho\xi u-d(u)}{\kappa-\rho\xi u+d(u)},
\quad\text{and}\quad
\beta_t  := \frac{\xi^2}{4\kappa}\left(1-\E^{-\kappa t}\right).
\end{equation}
The first step is to characterise the forward time-scale in the Heston model. 
In order to achieve this we first need to understand the limiting behaviour of the re-scaled $B$ function in~\eqref{eq:ABFunctions} 
that plays a fundamental role in the analysis below. 
The following lemma shows that using $h(\tau)\equiv\sqrt{\tau}$ as a time-scale produces the only non-trivial limit for the re-scaled $B$ function. We then immediately prove Lemma~\ref{Lemma:HestFwdTimeScale} which characterises the forward time-scale in the Heston model.

\begin{lemma}\label{lemma:Blimits}
Let $h:\RR_+\to\RR_+$ be a continuous function such that $\lim_{\tau\searrow 0}h(\tau)=0$ and $a\in\mathbb{R}_+^*$. 
The following limit then holds for $B$ in~\eqref{eq:ABFunctions} for all $u\in\mathbb{R}^*$:
$$
   \lim_{\tau\to 0}B(u/h(\tau),\tau)=\left\{ 
  \begin{array}{l l}
    \text{undefined}, & \quad \text{if}\quad \tau/h(\tau)\nearrow \infty,\\
+\infty, & \quad \text{if}\quad h(\tau)\equiv a\tau,\\
+\infty, & \quad \text{if}\quad \sqrt{\tau}/h(\tau)\nearrow \infty
\quad\text{and}\quad  \tau/h(\tau)\searrow  0,\\
     0, & \quad \text{if}\quad \sqrt{\tau}/h(\tau)\searrow  0,\\
u^2/(2a^2), & \quad \text{if}\quad h(\tau)\equiv a\tau^{1/2}.\\
\end{array} \right.
$$
\end{lemma}
\begin{proof}
As $\tau$ tends to zero we have the following asymptotic behaviours
for the functions $d$ and $\gamma$ defined in~\eqref{eq:DGammaBeta}:
\begin{equation}\label{eq:dgasymp}
\begin{array}{rl}
d\left(u/h(\tau)\right)
 & =\displaystyle \frac{1}{h(\tau)}\left(\kappa^2h(\tau)^2+u h(\tau)\left(\xi-2\kappa\rho\right)-\bar{\rho}^2\xi^2u^2\right)^{1/2}
 = \frac{\I u}{h(\tau)}d_0 + d_1 + \mathcal{O}(h(\tau)),  \\
\gamma\left(u/h(\tau)\right)
 & =\displaystyle \frac{\kappa h(\tau)-\rho\xi u-\I u d_0 - d_1h(\tau) + \mathcal{O}\left(h(\tau)^2\right)}{\kappa h(\tau)-\rho\xi u +\I u d_0 + d_1h(\tau) + \mathcal{O}\left(h(\tau)^2\right)}
= g_0-\frac{\I h(\tau)}{u}g_1+\mathcal{O}\left(h(\tau)^2\right), 
\end{array}
\end{equation}
where we have set
\begin{align}
d_0:= \bar{\rho}\xi\, \sgn(u),\qquad 
d_1:=\frac{ \I\left(2 \kappa  \rho -\xi \right)\sgn(u)}{2 \bar{\rho}},\qquad
g_0:= \frac{\I \rho -\bar{\rho}\,\sgn(u)}{\I\rho+\bar{\rho}\,\sgn(u)}\qquad
g_1:= \frac{\left(2 \kappa - \rho\xi\right) \sgn(u) }{\xi  \bar{\rho}\left(\bar{\rho }+\I \rho\,\sgn(u) \right)^2}.
\end{align}
First let $\tau/h(\tau)\to\infty$.
Then
$\exp\left({-d\left(u/h(\tau)\right)\tau}\right)=\exp\left({-\I \tau \bar{\rho}\xi|u|/h(\tau)+\mathcal{O}\left(\tau\right)}\right),$
and so the limit is undefined (complex infinity). 
Next let $\tau/h(\tau)\equiv 1/a.$ Using~\eqref{eq:dgasymp} we see that
$$
B(u/h(\tau),\tau) = 
-\left(\frac{u\rho+\I\bar{\rho}|u|}{\xi h\left(\tau\right)}\right)\frac{1-\E^{-\I\xi\bar{\rho}|u|/a}}{1-g_0  \E^{-\I\xi\bar{\rho}|u|/a}}+\mathcal{O}\left(1\right)
=a \zeta(u/a)/h(\tau)+\mathcal{O}\left(1\right),
$$	
where
$\zeta(u):=u\left(\bar{\rho}\xi\cot\left(u\xi\bar{\rho}/2\right)-\rho\xi\right)^{-1}$,
which is strictly positive for $u\in\mathbb{R}^*$ and $\zeta(0)=0$.
It follows that the limit in this case is infinite. 
Next let $\tau/h(\tau)\to0$. Here we can write
\begin{align}\label{eq:Bhasymp}
 \nonumber B(u/h(\tau),\tau) &= 
\left(-\frac{\rho u+\I\bar{\rho}|u|}{\xi h(\tau)}+\mathcal{O}\left(1\right)\right)\left(\left(\frac{1}{g_0 -1}+\mathcal{O}\left(h(\tau)\right)\right)\left(\frac{-\I\tau \bar{\rho}\xi|u|}{h(\tau)}+\mathcal{O}\left(\tau\right)\right)+\mathcal{O}\left(\left(\frac{\tau}{h(\tau)}\right)^2\right)\right) \\ \nonumber
&=\left(\frac{\rho u+\I\bar{\rho}|u|}{\xi}\right)\left(\frac{1}{g_0 -1}\right)\frac{\I \tau \bar{\rho}\xi|u|}{h(\tau)^2}+\mathcal{O}\left(\tau/h(\tau)\right) \\ 
&=\frac{u^2}{2}\left(\frac{\sqrt{\tau}}{h\left(\tau\right)}\right)^2+\mathcal{O}\left(\tau/h(\tau)\right).
\end{align}
If $\sqrt{\tau}/h(\tau)$ tends to infinity, so does $B(u/h(\tau),\tau)$. 
When $\sqrt{\tau}/h(\tau)$ tends to zero then $B(u/h(\tau),\tau)$ does as well.
If $\sqrt{\tau}/h(\tau)$ converges to a constant $1/a$, then $B(u/h(\tau),\tau)$ converges to~$u^2/(2a^2)$,
and the lemma follows.
\end{proof}

\begin{proof}[Proof of Lemma~\ref{Lemma:HestFwdTimeScale}]

For any $t>0$, the random variable $V_t$ in~\eqref{eq:Heston} is distributed as 
$\beta_t$ times a non-central chi-square random variable 
with $q=4\kappa\theta/\xi^2>0$ degrees of freedom 
and non-centrality parameter $\lambda=v \E^{-\kappa t}/\beta_t>0$. 
It follows that the corresponding moment generating function is given by
\begin{equation}\label{eq:HestonVarianceMGF}
\Lambda_{t}^{V}(u)
:=\mathbb{E}\left(\E^{uV_t}\right)
=\exp\left({\frac{\lambda \beta_tu}{1-2\beta_tu}}\right)\left(1-2\beta_tu\right)^{-q/2},
\qquad\text{for all }u<\frac{1}{2\beta_t}.
\end{equation}
The re-normalised Heston forward logarithmic moment generating function is then computed as
\begin{align*}
\Lambda_{\tau}^{(t)}(u,h(\tau))/h(\tau)
&=\log\mathbb{E}\left[\E^{u\left(X_{ t+\tau}-X_{t}\right)/h(\tau)}\right]
= \log\mathbb{E}\left[\mathbb{E}\left(\E^{u\left(X_{ t+\tau}
-X_{t}\right)/h(\tau)}|\mathcal{F}_{ t}\right)\right] \\ \nonumber
&=\log\mathbb{E}\left(\E^{A\left(u/h(\tau),\tau\right)
+B\left(u/h(\tau),\tau\right)V_{ t}}\right)
=A\left(u/h(\tau),\tau\right) + \log\Lambda_{t}^{V}\left(B\left(u/h(\tau),\tau\right)\right),\nonumber
\end{align*}
which agrees with~\eqref{eq:LambdaTau} when $h(\tau)\equiv 1$.
This is only valid in some effective domain~$\mathcal{D}_{t,\tau}\subset\mathbb{R}$.
The lmgf for $V_{ t}$ is well defined in
$\mathcal{D}_{t,\tau}^{V}
:=\{u\in\mathbb{R}:B\left(u/h(\tau),\tau\right)<1/(2\beta_{ t})\}$, 
and hence 
$\mathcal{D}_{t,\tau}
=\mathcal{D}_{ t,\tau}^{V} \cap \mathcal{D}_{\tau}$, 
where
$\mathcal{D}_{\tau}$ is the effective domain of the (spot) re-normalised Heston lmgf.
Consider first $\mathcal{D}_{\tau}$ for small $\tau$. 
From~\cite[Proposition 3.1]{AP07} if $\xi ^2 (u/h(\tau)-1) u/h(\tau)>(\kappa - \rho\xi u/h(\tau))^2$ 
then the explosion time~$\tau^*(u):=\sup\{t\geq 0:\mathbb{E}(\E^{u X_{t}})<\infty\}$ of the Heston mgf is
\begin{align*}
\tau^*_H\left(u/h(\tau)\right)
&=\frac{2}{\sqrt{\xi ^2 (u/h(\tau)-1) u/h(\tau)-(\kappa-\rho\xi  u/h(\tau))^2}}
\Big\{\pi\ind_{\{\rho\xi u/h(\tau)-\kappa<0\}} \Big. \\ \nonumber
& \Big.+\arctan\left(\frac{\sqrt{\xi ^2 (u/h(\tau)-1) u/h(\tau)-(\kappa-\rho\xi  u/h(\tau))^2}}{\rho\xi u/h(\tau)-\kappa }\right)\Big\}.
\end{align*}
Recall the following Taylor series expansions, for $x$ close to zero:
\begin{equation*}
\left.
\begin{array}{rll}
\displaystyle\arctan\left(\frac{1}{\rho\xi u/x-\kappa }\sqrt{\xi^2 \left(\frac{u}{x}-1\right)
 \frac{u}{x}-\left(\kappa -\xi  \rho  \frac{u}{x}\right)^2}\right)
 & =\displaystyle\sgn(u)\arctan\left(\frac{\bar{\rho}}{  \rho  }\right)+\mathcal{O}\left(x\right),
\quad & \text{if } \rho\ne 0,\\
\displaystyle\arctan\left(-\frac{1}{\kappa}\sqrt{\xi ^2 \left(\frac{u}{x}-1\right) \frac{u}{x}-\kappa^2}\right)
 & =\displaystyle-\frac{\pi}{2}+\mathcal{O}(x),
\quad & \text{if } \rho = 0.
\end{array}
\right.
\end{equation*}
As $\tau$ tends to zero 
$\xi ^2 (u/h(\tau)-1) u/h(\tau)>(\kappa-\rho\xi u/h(\tau))^2$ is satisfied since
$\rho^2<1$ and hence
$$
\tau^*_H\left(u/h(\tau)\right)=\left\{ 
  \begin{array}{l l}
\displaystyle   \frac{h(\tau)}{\xi |u|}\left(\pi\ind_{\{\rho=0\}}
+\frac{2}{\bar{\rho}}\left(\pi\ind_{\{\rho u \leq0\}}+ \sgn(u)\arctan\left(\frac{\bar{\rho}}{\rho}\right)\right)\ind_{\{\rho\neq0\}}+\mathcal{O}\left(h(\tau)\right)\right), 
& \quad \text{if } u \neq 0,\\
  \infty, & \quad \text{if }u=0.
    \\
\end{array} \right.
$$
Therefore, for $\tau$ small enough, we have $\tau^*_H\left(u/h(\tau)\right)>\tau$ for all
$u\in\mathbb{R}$ if $\tau/h\left(\tau\right)$ tends to zero and $\tau^*_H\left(u/h(\tau)\right)>\tau$ for all $u\in\left(u_{-},u_{+}\right)$ if $h\left(\tau\right)\equiv a\tau$, where 
\begin{align*}
u_{-}
 & :=\frac{2a}{\bar{\rho}\xi\tau}\arctan\left(\frac{\bar{\rho}}{\rho}\right)\ind_{\{\rho<0\}}
-\frac{\pi a}{\xi\tau}\ind_{\{\rho=0\}}
+\frac{2 a}{\bar{\rho}\xi\tau}\left(\arctan\left(\frac{\bar{\rho}}{\rho}\right)-\pi\right)\ind_{\{\rho>0\}},\\
u_{+}
 & :=\frac{2 a}{\bar{\rho}\xi\tau}\left(\arctan\left(\frac{\bar{\rho}}{\rho}\right)
+\pi\right)\ind_{\{\rho<0\}}+\frac{\pi a}{\xi\tau}\ind_{\{\rho=0\}}
+\frac{2 a}{\bar{\rho}\xi\tau}\arctan\left(\frac{\bar{\rho}}{\rho}\right)\ind_{\{\rho>0\}}.
\end{align*}
If $\tau/h(\tau)$ tends to infinity, then $\tau^*_H\left(u/h(\tau)\right)\leq\tau$ for all~$u\in\mathbb{R}^*$. 
We are also required to find $\mathcal{D}_{ t,\tau}^{V}$ for small~$\tau$. 
Using Lemma~\ref{lemma:Blimits} we see 
that if $h(\tau)\equiv a\tau^{1/2}$ then $\lim_{\tau\searrow 0} \mathcal{D}_{t,\tau}^{V}=\{u\in\mathbb{R}:|u|<a/{\sqrt{\beta_t}}\}$. 
By the limit of a set we precisely mean the following:
$$
\liminf_{\tau\searrow 0}\mathcal{D}_{ t,\tau}^{V}
:= \bigcup_{\tau>0}\bigcap_{s\leq\tau}\mathcal{D}_{ t,s}^{V}
 = \bigcap_{\tau>0}\bigcup_{s\leq\tau}\mathcal{D}_{ t,s}^{V}
=:\limsup_{\tau\searrow 0}\mathcal{D}_{ t,\tau}^{V}.
$$
If $\tau^{1/2}/h(\tau)$ tends to infinity then $\lim_{\tau\searrow 0}\mathcal{D}_{t,\tau}^{V}=\{0\}$ and if it tends to zero, then $\lim_{\tau\searrow 0}\mathcal{D}_{t,\tau}^{V}=\mathbb{R}$. 
The limiting domains in the lemma follow after taking the appropriate intersections. 
Next we move on to the limits. 
We only consider the cases where $h(\tau)\equiv a\tau^{1/2}$
and where $\tau^{1/2}/h(\tau)$ tends to zero since these are the only cases for which the forward logarithmic moment generating function is defined. 
Using~\eqref{eq:Bhasymp} we see as $\tau$ tends to zero
\begin{equation}\label{eq:fwdmgf1asymp}
  \log\left(1-2\beta_t B\left(u/h(\tau),\tau\right)\right)=\frac{B(u/h(\tau),\tau)v\E^{-\kappa t}}{1-2\beta_t B(u/h(\tau),\tau)}=\left\{ 
  \begin{array}{l l}
    \mathcal{O}(1), & \quad \text{if}\quad h(\tau)\equiv a\tau^{1/2},\\
\mathcal{O}(\tau/h(\tau)), & \quad \text{if}\quad \sqrt{\tau}/h(\tau) \searrow 0.\\
\end{array} \right.
\end{equation}
The lemma follows from this and the fact that the function $A$ in~\eqref{eq:ABFunctions} satisfies
$A(u/h(\tau),\tau)=\mathcal{O}\left((\tau/h(\tau))^2\right)$.
\end{proof}

\subsection{Asymptotic time-dependent measure-change }\label{sec:asympmeaschange}

In this section we define the fundamental asymptotic time-dependent measure-change in~\eqref{eq:MeasureChangeSmallTime} and derive expansions for critical functions related to this measure-change. 
In order to proceed with this program we first need to prove some technical lemmas.
We use our forward time-scale and define the following rescaled quantities:
\begin{equation}\label{eq:RenLamAB}
\Lambda^{(t)}_\tau (u):=\Lambda^{(t)}_\tau (u,\sqrt{\tau}),\quad \widehat{A}(u):=A(u/\sqrt{\tau},\tau),\quad
\widehat{B}(u):=B(u/\sqrt{\tau},\tau),
\end{equation}
with $\Lambda^{(t)}_\tau$, $A$ and $B$ defined in~\eqref{eq:LambdaTau} and~\eqref{eq:ABFunctions} respectively. The following lemma gives the asymptotics of the re-scaled quantities $\widehat{A}$, $\widehat{B}$ as $\tau$ tends to zero:

\begin{lemma}\label{lemma:ABHat}
The following expansions hold for all $u\in\mathcal{D}_{\Lambda}$ as $\tau$ tends to zero
($\widehat{B}_1$ was defined in~\eqref{eq:B1hat}):
\begin{equation}
\widehat{B}(u)=\frac{u^2}{2}+\widehat{B}_1(u)\sqrt{\tau}+\mathcal{O}(\tau),\qquad 
\widehat{A}(u)=\frac{u^2\kappa\theta\tau}{4}+\mathcal{O}(\tau^{3/2}).
\end{equation}
\end{lemma}

\begin{proof}
From the definition of $A$ in~\eqref{eq:ABFunctions} and the asymptotics in~\eqref{eq:dgasymp} with $h(\tau)\equiv\sqrt{\tau}$ we obtain
\begin{align} 
\widehat{A}(u) := A\left(u/\sqrt{\tau},\tau\right) 
 & = 
\frac{\kappa\theta}{\xi^2}\left(\left(\kappa-\frac{\rho\xi u}{\sqrt{\tau}}- d\left(\frac{u}{\sqrt{\tau}}\right)\right)\tau-2\log\left(\frac{1-\gamma(u/\sqrt{\tau})
\exp\left(-d(u/\sqrt{\tau})\tau\right)}{1-\gamma(u/\sqrt{\tau})}\right)\right) \\ \nonumber
 & = \frac{\kappa\theta}{\xi^2}
\left(\left(\kappa-\frac{\rho\xi u}{\sqrt{\tau}}-\frac{\I u d_0}{\sqrt{\tau}}-d_1 + \mathcal{O}(\sqrt{\tau})\right)\tau \right.\\ \nonumber
& \left. -2\log\left(\frac{1-\left(g_0-\I\sqrt{\tau} g_1/u+\mathcal{O}(\tau)\right)\exp\left(-\I u d_0 \sqrt{\tau} -  d_1 \tau +\mathcal{O}(\tau^{3/2})\right)}{1-\left(g_0-\I\sqrt{\tau} g_1/u+\mathcal{O}(\tau)\right)}\right)\right) \\ \nonumber
&=u^2\theta\kappa\tau/4+\mathcal{O}(\tau^{3/2}).
\end{align}
Substituting the asymptotics for $d$ and $\gamma$ in ~\eqref{eq:dgasymp}
we further obtain
$$
\frac{1-\exp\left(-d(u/\sqrt{\tau})\tau\right)}
{1-\gamma(u/\sqrt{\tau})\exp\left(-d(u/\sqrt{\tau})\tau\right)}
 = \frac{1-\exp\left(-\I u d_0 \sqrt{\tau} -  d_1 \tau +\mathcal{O}(\tau^{3/2})\right)}{1-\left(g_0-\I\sqrt{\tau} g_1/u+\mathcal{O}(\tau)\right)
\exp\left(-\I u d_0 \sqrt{\tau} -  d_1 \tau +\mathcal{O}(\tau^{3/2})\right)},
$$
and therefore using the definition of $B$ in~\eqref{eq:ABFunctions} we obtain
\begin{align}
\widehat{B}(u) := B\left(\frac{u}{\sqrt{\tau}},\tau\right) 
 &= \frac{\kappa-\rho\xi u/\sqrt{\tau}-d(u/\sqrt{\tau})}{\xi^2}\frac{1-\exp\left(-d\left(u/\sqrt{\tau}\right)\tau\right)}{1-\gamma\left(u/\sqrt{\tau}\right)\exp\left(-d\left(u/\sqrt{\tau}\right)\tau\right)} \\ \nonumber
 & = -\frac{\rho\xi u+\I u d_0}{\xi^2}\frac{\I u d_0}{1-g_0} + \widehat{B}_1(u)\sqrt{\tau}+\mathcal{O}(\tau)
=\frac{u^2}{2}+\widehat{B}_1(u)\sqrt{\tau}+\mathcal{O}(\tau).
\end{align}
\end{proof}

It is still not clear what benefit the forward time-scale has given us since the limiting lmgf is still degenerate.
Firstly, even though the limiting lmgf is zero on a bounded interval, the re-scaled forward lmgf for fixed $\tau>0$ is still steep on the domain of definition which implies the existence of a unique solution $u^*_{\tau}(k)$ to the equation 
\begin{equation}\label{eq:u^*tausmalltime}
\partial_{u}\Lambda^{(t)}_\tau  (u^*_{\tau}(k))=k.
\end{equation}
Further as $\tau$ tends to zero, $u^*_{\tau}(k)$ converges to $1/\sqrt{\beta_t}$ when $k>0$ 
and to $-1/\sqrt{\beta_t}$ when $k<0$ (see Lemma~\ref{lemma:u^*tausmalltime} below).
The key observation is that the forward time-scale ensures finite boundary points for the effective domain,
which in turn implies finite limits for $u^*_{\tau}(k)$. 
This is critical to the asymptotic analysis that follows and it will become clear that if any other time-scale were to be used the analysis would break down. 
The following lemma shows that our definition~\eqref{eq:u^*tausmalltime} of $u^*_{\tau}(k)$ is exactly what we need to conduct an asymptotic analysis in this degenerate case.
\begin{lemma} \label{lemma:u^*tausmalltime}
For any $k\in\RR$, $\tau>0$, the equation~\eqref{eq:u^*tausmalltime} 
admits a unique solution~$u^*_{\tau}(k)$;
as $\tau$ tends to zero, it converges to $1/\sqrt{\beta_t}$ ($-1/\sqrt{\beta_t}$) when $k>0$ ($k<0$),
to zero when~$k=0$,
and $u^*_{\tau}(k)\in\mathcal{D}_{\Lambda}^{o}$ for~$\tau$ small enough.
\end{lemma}


\begin{proof}
We first start by the following claims, which can be proved using the convexity of the forward moment generating function and tedious computations;
we shall not however detail these lengthy computations here for brevity, but Figure~\ref{fig:MGFConv} below provides a visual help.
\begin{enumerate}[(i)]
\item For any $\tau>0$, the map $\partial_u \Lambda_{\tau}^{(t)}:\mathcal{D}_{t,\tau}\to\mathbb{R}$ is strictly increasing
and the image of 
$\mathcal{D}_{t,\tau}$ by $\partial_u \Lambda_{\tau}^{(t)}$ is $\mathbb{R}$;
\item For any $\tau>0$, $u_{\tau}^*(0)>0$  and $\lim_{\tau\searrow0}u_{\tau}^*(0)=0$,
i.e. the unique minimum of $\Lambda_{\tau}^{(t)}$ converges to zero;
\item For each $u\in\mathcal{D}_{\Lambda}^{o}$, $\partial_u \Lambda_{\tau}^{(t)}(u)$ converges to zero 
as $\tau$ tends to zero.
\end{enumerate}
Now, choose $k>0$ (analogous arguments hold for $k<0$).
It is clear from~(i) that ~\eqref{eq:u^*tausmalltime} admits a unique solution.
Since $\lim_{\tau\searrow0}\mathcal{D}_{t,\tau}=\mathcal{D}_{\Lambda}$, then there exists $\tau_1>0$ such that for any $\tau<\tau_1$, 
$u_{\tau}^*(k)\in\mathcal{D}_{\Lambda}^{o}$.
Note further that (i) and (ii) imply $u_{\tau}^*(k)>0$.
From~(iii) there exists $\tau_2>0$ such that the sequence $(u_{\tau}^*(k))_{\tau>0}$ 
is strictly increasing as $\tau$ goes to zero for $\tau<\tau_2$.
Now let $\tau^*=\min(\tau_1,\tau_2)$ and consider $\tau<\tau^*$.
Then $u_{\tau}^*(k)$ is bounded above by $1/\sqrt{\beta_t}$ (because $u_{\tau}^*(k)\in\mathcal{D}_{\Lambda}^{o}$) 
and therefore converges to a limit $L\in[0,1/\sqrt{\beta_t}]$. 
Suppose that  $L\ne1/\sqrt{\beta_t}$. 
Since $s\mapsto u_s^*(k)$ is strictly increasing as $s$ tends to zero (and $s<\tau^*$), 
and $\partial_u \Lambda_{\tau}^{(t)}$ is strictly increasing we have $\partial_u \Lambda_{\tau}^{(t)}(u_\tau^*(k))\leq \partial_u \Lambda_{\tau}^{(t)}(L)$;
Combining this and~(iii) yields
$$
\lim_{\tau\searrow 0}\partial_u \Lambda_{\tau}^{(t)}(u_\tau^*(k)) \leq  
\lim_{\tau\searrow 0}\partial_u \Lambda_{\tau}^{(t)}(L) = 0
\neq k,
$$
which contradicts the assumption $k>0$. 
Therefore $L=1/\sqrt{\beta_t}$ and the lemma follows.
\begin{figure}[h!tb] 
\centering
\mbox{{\includegraphics[scale=0.7]{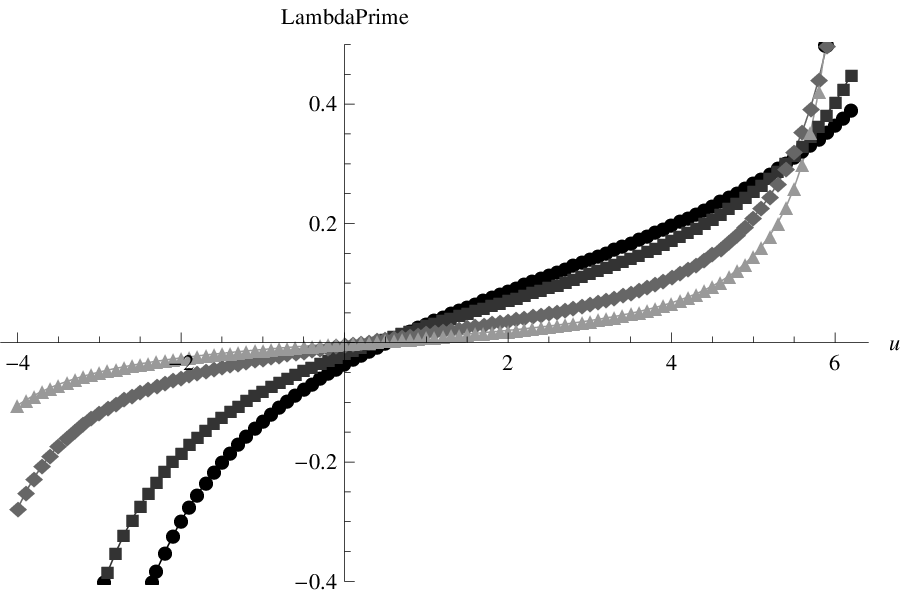}}\quad
{\includegraphics[scale=0.7]{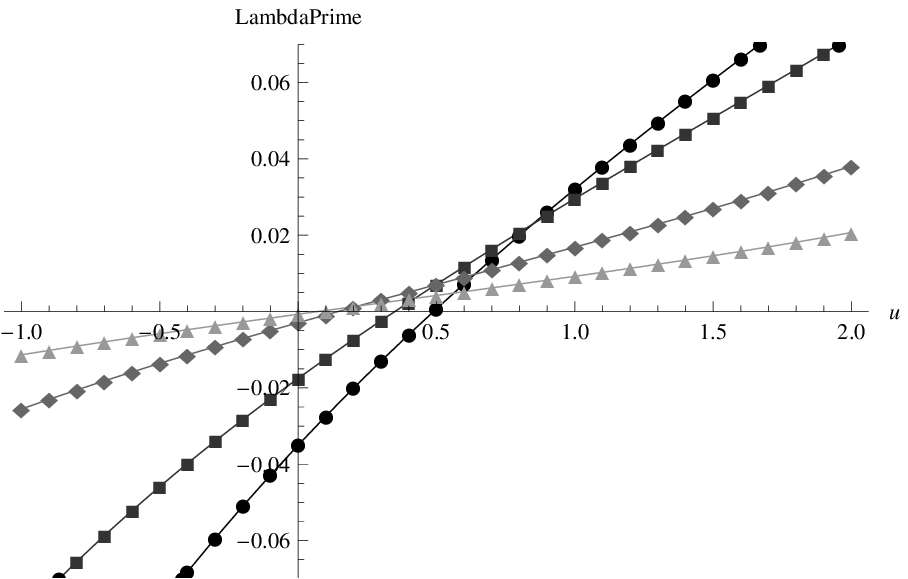}}}
\caption{Plot of $u\mapsto \partial_u \Lambda_{\tau}^{(t)}(u)$ for different values of $\tau$. 
Circles, squares, diamonds and triangles represent $\tau=1,1/2,1/12,1/50$. 
The forward-start date is $t=1$ and the Heston model parameters are $v=\theta=0.07$, $\xi=0.4$, $\rho=-0.6$, $\kappa=1$. 
The limiting domain is $(-1/{\sqrt{\beta_t}},1/{\sqrt{\beta_t}})\approx(-6.29,6.29)$. 
The right plot is a zoomed version of the left graph.}
\label{fig:MGFConv}
\end{figure}
\end{proof}

For sufficiently small $\tau$ we introduce a time-dependent change of measure by
\begin{equation}\label{eq:MeasureChangeSmallTime}
\frac{\D\mathbb{Q}_{k,\tau}}{\D\mathbb{P}}:=\exp\left( u^*_{\tau}(k)X_{\tau}^{(t)}/\sqrt{\tau}-\Lambda^{(t)}_{\tau}(u^*_{\tau}(k))/\sqrt{\tau}\right).
\end{equation}
By Lemma~\ref{lemma:u^*tausmalltime}, $u^*_\tau(k)\in\mathcal{D}_{\Lambda}^0$ for $\tau$ small enough 
and so $|\Lambda^{(t)}_{\tau}(u^*_\tau)|$ is finite since $\mathcal{D}_{\Lambda}=\lim_{\tau\searrow 0} \{u\in\mathbb{R}:|\Lambda_{\tau}^{(t)}(u)|<\infty\}$. 
Also $\D\mathbb{Q}_{k,\tau}/\D\mathbb{P}$ is almost surely strictly positive and by definition $\mathbb{E}[\D\mathbb{Q}_{k,\tau}/\D\mathbb{P}]=1$. 
Therefore~\eqref{eq:MeasureChangeSmallTime} is a valid measure change for all $k\in\mathbb{R}^*$ and sufficiently small $\tau$. 
Equation~\eqref{eq:u^*tausmalltime} can be written explicitly as 
\begin{equation}\label{eq:fundeqnsmalltime}
\frac{\sqrt{\tau} \E^{-\kappa t}}{k \xi^2} \Big[\xi ^2 \E^{\kappa  t} \widehat{A}'(u^*_{\tau}) \left(1-2 \widehat{B}(u^*_{\tau}) \beta _t\right)^2+\widehat{B}'(u^*_{\tau}) 
\left(4\kappa\theta 
	   \beta _t \E^{\kappa  t}(1-2\widehat{B}(u^*_{\tau}) \beta _t)+\xi ^2 v\right)\Big]
=\left(1 -2 \widehat{B}(u^*_{\tau}) \beta _t\right)^2,
\end{equation}
with $\widehat{A}$ and $\widehat{B}$ defined in~\eqref{eq:RenLamAB}. 
We now use this to derive an asymptotic expansion for $u^*_{\tau}$ as $\tau$ tends to zero.
\begin{lemma}\label{Lemma:ustartausmalltime}
The expansion 
$u^*_\tau(k)=a_0(k) + a_1(k)\tau^{1/4}+a_2(k)\tau^{1/2}+a_3(k)\tau^{3/4}+\mathcal{O}(\tau)$
holds for all $k\in\mathbb{R}^*$ as~$\tau$ tends to zero,
with  $a_0$, $a_1$, $a_2$ and $a_3$ defined in~\eqref{eq:aB1hatsmalltime}.
\end{lemma}

\begin{proof}
Existence and uniqueness was proved in Lemma~\ref{lemma:u^*tausmalltime} and so we assume the result as an ansatz.
Consider $k>0$.
From Lemma~\ref{lemma:u^*tausmalltime} it is clear that $a_0(k)=1/\sqrt{\beta_t}$.
The ansatz and Lemma~\ref{lemma:ABHat} then imply the following asymptotics as $\tau$ tends to zero 
(we drop here the $k$-dependence):
\begin{equation}\label{eq:AprimeBBprimeAsymp}
\begin{array}{ll}
\widehat{B}(u^*_{\tau}) & = \displaystyle \frac{1}{2 \beta _t}
+a_0 a_1\tau^{1/4}
+ r\tau^{1/2} +\left(a_1a_2+a_0 a_3+a_1\widehat{B}_1'(a_0)\right)\tau^{3/4}+\mathcal{O}(\tau), \\
\widehat{B}'(u^*_{\tau}) & = \displaystyle a_0+a_1\tau^{1/4}+(a_2+\widehat{B}_1'(a_0))\tau^{1/2}+\mathcal{O}(\tau^{3/4}), \\
\widehat{A}'(u^*_{\tau}) & = \displaystyle \frac{1}{2}\kappa\theta a_0 \tau+\mathcal{O}(\tau^{5/4}),
\end{array}
\end{equation}
where $r\equiv r(k):=a_0 a_2+\widehat{B}_1(a_0)+a_1^2/2
 =a_1^2/2-\kappa\theta /(|k|\xi^2\sqrt{\beta_t})$ is defined in~\eqref{eq:sigmarsmalltime}.
We substitute these asymptotics into~\eqref{eq:fundeqnsmalltime} and solve at each order. 
At the $\tau^{1/4}$ order we have two solutions, 
$
 a_1(k)=\pm \sqrt{v} \E^{-\kappa  t/2}/(2 \sqrt{k} \beta _t^{3/4})
$
and we choose the negative root so that $u^*_{\tau}\in\mathcal{D}^o_{\Lambda}$ for $\tau$ small enough. 
In a straightforward, yet tedious, manner we continue the procedure and iteratively solve at each order 
(the next two equations are linear in $a_2$ and $a_3$) to derive the asymptotic expansions in the lemma.
An analogous treatment holds in the case $k<0$.

To complete the proof (and make the ansatz approach above rigorous) we need to show the existence of this expansion for $u_{\tau}^*(k)$.
Fix $k\in\mathbb{R}^*$ and set $f_k(u,\tau):=\partial_u\Lambda_{\tau}^{(t)}(u)-k$.
Now let $\bar{\tau}>0$.
From Lemma~\ref{lemma:u^*tausmalltime} we know that there exists a solution $u_{\bar{\tau}}^*(k)$ to the equation $f_k(u_{\bar{\tau}}^*(k),\bar{\tau})=0$ and 
the strict convexity of the forward lmgf implies $\partial_u f_k(u_{\bar{\tau}}^*(k),\bar{\tau})>0$.
Further, the two-dimensional map $f_k:\mathcal{D}_{t,\tau}^{o}\times\mathbb{R}^*_{+}\to\mathbb{R}$ is analytic. 
It follows by the Implicit Function Theorem~\cite[Theorem 8.6, Chapter 0]{Kaup83} 
that $\tau\mapsto u^*_{\tau}(k)$ is analytic in some neighbourhood around $\bar{\tau}$.
Since this argument holds for all $\bar{\tau}>0$,  this function is also analytic on $\mathbb{R}^*_{+}$.
Also from Lemma~\ref{lemma:u^*tausmalltime} we know that 
$\lim_{\tau\searrow 0}u^*_{\tau}(k) = \sgn(k)/\sqrt{\beta_t}$.
Since we computed the Taylor series expansion consistent with this limit and the expansion is unique, it follows that $u^*_{\tau}(k)$ admits this representation.
\end{proof}
In the forthcoming analysis we will be interested in the asymptotics of 
\begin{equation} \label{eq:etausmalltime}
e_{\tau}(k):=\left(1-2 \widehat{B}(u^*_{\tau}(k)) \beta _t\right)\tau^{-1/4},
\end{equation}
as $\tau$ tends to zero. 
Since $(1-2 \widehat{B}(u^*_{\tau}(k)) \beta _t)$ converges to zero,
it is not immediately clear that $e_{\tau}$ has a well defined limit. 
But we can use the asymptotics in~\eqref{eq:AprimeBBprimeAsymp} to deduce the following lemma:
\begin{lemma}\label{Lemma:etausmalltime}
The expansion 
$e_{\tau}(k)=e_0(k)+e_1(k)\tau^{1/4}+e_2(k)\tau^{1/2}+\mathcal{O}(\tau^{3/4})$
holds for all $k\in\mathbb{R}^*$ as $\tau$ tends to zero,
where $e_0$, $e_1$ and $e_2$ are defined in~\eqref{eq:esmalltime}.
\end{lemma}
\begin{proof}
We substitute the asymptotics for $\widehat{B}(u^*_{\tau})$ in~\eqref{eq:AprimeBBprimeAsymp} into the definition of $e_\tau$ in~\eqref{eq:etausmalltime} and the lemma follows after simplification.
\end{proof}

\subsection{Characteristic function asymptotics} \label{sec: Charactfuncasymp}

We now define the random variable $\label{eq:DefZSmallTime} Z_{\tau,k}:=\left(X_{\tau}^{(t)}- k\right)/\tau^{1/8}$ 
and the characteristic function $\Phi_{\tau,k}:\mathbb{R}\to\mathbb{C}$ of $Z_{\tau,k}$ in the $\mathbb{Q}_{k,\tau}$-measure  in~\eqref{eq:MeasureChangeSmallTime} as 
\begin{equation}\label{eq:CharacNonSteepSmallTime}
\Phi_{\tau,k}(u):=\mathbb{E}^{\mathbb{Q}_{k,\tau}}\left(\E^{\I u Z_{\tau,k}}\right).
\end{equation}
Define now the functions $\phi_1,\phi_2:\mathbb{R}^*\times\mathbb{R}\to\mathbb{C}$ by
$$
\phi_1(k,u) := \I u \left(\psi _0(k)+\frac{4 a_0(k)\theta  \kappa \beta _t}{e_0(k) \xi ^2}\right)+\I u^3 \psi_1(k),
\qquad\text{and}\qquad
\phi_2(k,u) := u^2\phi_2^{a}(k)+u^4\phi_2^{b}(k)+u^6\phi_2^{c}(k),
$$
with $\psi_0$, $\psi_1$ defined in~\eqref{eq:psismalltime}, $a_0$, $e_0$ in~\eqref{eq:aB1hatsmalltime},~\eqref{eq:esmalltime},
and $\phi_2^{a}$, $\phi_2^{b}$ and $\phi_2^{c}$ in~\eqref{eq:phi2smalltime}.
The following lemma provides the asymptotics of~$\Phi_{\tau,k}$:
\begin{lemma}\label{Lemma:charactfuncasympsmalltime}
The following expansion holds for all $k\in\mathbb{R}^*$ as $\tau$ tends to zero
(with $\zeta$ given in~\eqref{eq:sigmarsmalltime}):
$$\Phi_{\tau,k}(u)=\E^{-\frac{1}{2}\zeta^2(k) u^2}\left(1+\phi_1(k,u)\tau^{1/8}+\phi_2(k,u)\tau^{1/4}+\mathcal{O}\left(\tau^{3/8}\right)\right).$$
\end{lemma}
\begin{remark}\label{rem:weakconsmalltau}
For any $k\in\RR^*$, L\'evy's Convergence Theorem~\cite[Page 185, Theorem 18.1]{W03} implies that
$Z_{\tau,k}$ converges weakly to a normal random variable with zero mean and variance $\zeta^2(k)$ as $\tau$ tends to zero.
\end{remark}
\begin{proof}
From the change of measure~\eqref{eq:MeasureChangeSmallTime} and the forward logarithmic moment generating function 
given in~\eqref{eq:LambdaTau} we compute (we drop the $k$-dependence throughout) for small~$\tau$:
\begin{align}
\log\Phi_{k,\tau}(u)
 & = \log\mathbb{E}^{\mathbb{P}}\left(\frac{\D\mathbb{Q}_{k,\tau}}{\D\mathbb{P}}\E^{\I u Z_{k,\tau}}\right)
=
\log\mathbb{E}^{\mathbb{P}}\left[\exp\left(\frac{u^*(k) X_{\tau}^{(t)}}{\sqrt{\tau}}
-\frac{\Lambda_{\tau}^{(t)}(u^*_{\tau})}{\sqrt{\tau}}\right)
\exp\left({ \frac{\I uX_{\tau}^{(t)}}{\tau^{1/8}}-\frac{\I uk}{\tau^{1/8}}}\right)\right] \nonumber \\ \nonumber
 & = {-\frac{1}{\sqrt{\tau}}\Lambda_{\tau}^{(t)}\left(u^*_{\tau}\right)-\frac{\I uk}{\tau^{1/8}}}+\log\mathbb{E}^{\mathbb{P}}\left[\exp\left({\left(\frac{X_{\tau}^{(t)}}{\sqrt{\tau}}\right)\left(\I u\tau^{3/8}+u^*_{\tau}\right)}\right)\right]  \\ \label{eq:phiasympsmalltime}
 & = -\frac{\I uk}{\tau^{1/8}}+\frac{1}{\sqrt{\tau}}\left(\Lambda_{\tau}^{(t)}\left(\I u\tau^{3/8} +u^*_{\tau}\right)-\Lambda_{\tau}^{(t)}\left(u^*_{\tau}\right)\right).
\end{align}
Using the asymptotics in~\eqref{eq:AprimeBBprimeAsymp} we have that as $\tau$ tends to zero (we drop the $k$-dependence)
\begin{align}\label{eq:BIasymp}
\widehat{B}\left(u^*_{\tau}+\I u\tau^{3/8}\right)&=\frac{a_0^2}{2}+a_0 a_1 \tau^{1/4}+\I a_0 \tau ^{3/8} u+ r\tau^{1/2}+\mathcal{O}(\tau^{5/8}), \\ \nonumber
\widehat{B}\left(u^*_{\tau}\right)&=\frac{a_0^2}{2}+a_0 a_1 \tau^{1/4}+ r\tau^{1/2}+\mathcal{O}(\tau^{3/4}), \\ \nonumber
\widehat{B}\left(u^*_{\tau}+\I u\tau^{3/8}\right)-\widehat{B}\left(u^*_{\tau}\right)
 & = \I a_0 u \tau ^{3/8} +\I a_1 u \tau^{5/8} -\frac{1}{2}  u^2 \tau ^{3/4} +\mathcal{O}(\tau^{7/8}),
\end{align}
where $r\equiv r(k):=a_0 a_2+\widehat{B}_1(a_0)+a_1^2/2
 =a_1^2/2-\kappa\theta /(|k|\xi^2\sqrt{\beta_t})$ is defined in~\eqref{eq:sigmarsmalltime}.
Similarly for small $\tau$,
\begin{align}\label{eq:AIAsymp}
\widehat{A}\left(u^*_{\tau}+\I u\tau^{3/8}\right)
 &= \frac{\kappa\theta}{4} a_0^2\tau
+\frac{\kappa\theta}{2}a_0 a_1\tau^{5/4}+\frac{\I\kappa\theta}{2}a_0 u\tau ^{11/8} +\mathcal{O}(\tau^{3/2}), \\ \nonumber
\widehat{A}\left(u^*_{\tau}\right)
& = \frac{\kappa\theta}{4} a_0^2\tau +\frac{\kappa\theta}{2}a_0a_1\tau^{5/4}+\mathcal{O}(\tau^{3/2}), \\ \nonumber
\widehat{A}\left(u^*_{\tau}+\I u\tau^{3/8}\right)-\widehat{A}\left(u^*_{\tau}\right)
&= \frac{\I\kappa\theta}{2}a_0 u\tau^{11/8} +\mathcal{O}(\tau^{3/2}).
\end{align}
We now use $e_{\tau}$ defined in~\eqref{eq:etausmalltime} to re-write the term
\begin{equation*}
\frac{\widehat{B}\left(u^*_{\tau}+\I u\tau^{3/8}\right)v\E^{-\kappa t}}{1-2\beta_t\widehat{B}\left(u^*_{\tau}+\I u\tau^{3/8}\right)}
=\frac{v\E^{-\kappa t}\tau^{-1/4}\widehat{B}\left(u^*_{\tau}+\I u\tau^{3/8}\right)}{e_\tau-2\beta_t\tau^{-1/4}\left(\widehat{B}\left(u^*_{\tau}+\I u\tau^{3/8}\right)-\widehat{B}\left(u^*_{\tau}\right)\right)},
\end{equation*}
and then use the asymptotics in~\eqref{eq:BIasymp} and Lemma~\ref{Lemma:etausmalltime} to find that for small $\tau$
\begin{align}\label{eq:BratioIAsymp}
&\frac{v\E^{-\kappa t}\tau^{-1/4}\widehat{B}\left(u^*_{\tau}+\I u\tau^{3/8}\right)}{e_\tau-2\beta_t\tau^{-1/4}\left(\widehat{B}\left(u^*_{\tau}+\I u\tau^{3/8}\right)-\widehat{B}\left(u^*_{\tau}\right)\right)} \\ \nonumber
&=\frac{v\E^{-\kappa t}\tau^{-1/4}\left(a_0^2/2+a_0 a_1 \tau^{1/4}+\I a_0 \tau ^{3/8} u+ r\tau^{1/2}+\mathcal{O}(\tau^{5/8})\right)}{e_0+e_1\tau^{1/4}+e_2\tau^{1/2}+\mathcal{O}(\tau^{3/4})-2\beta_t\tau^{-1/4}\left(\I a_0 u \tau ^{3/8} +\I a_1 u \tau ^{5/8} -\frac{1}{2}  u^2 \tau ^{3/4} +\mathcal{O}(\tau^{7/8})\right)} \\ \nonumber
&=\frac{v\E^{-\kappa t}a_0^2}{2 e_0}\tau^{-1/4}+\frac{v\E^{-\kappa t} \I a_0^3 u \beta _t}{e_0^2}\tau^{-1/8}
+v\E^{-\kappa t}\left(\frac{a_0 a_1}{e_0}-\frac{a_0^2 e_1}{2 e_0^2}\right)  \\ \nonumber
&-\frac{\zeta^2 u^2}{2}+(\I u \psi_0 +\I u^3 \psi_1)\tau^{1/8}+(\psi_4+\psi_2 u^2 +\psi_3 u^4)\tau^{1/4}+\mathcal{O}(\tau^{3/8}),
\end{align}
with $\zeta$ and $\psi_0,\ldots,\psi_4$  defined in~\eqref{eq:sigmarsmalltime} and~\eqref{eq:psismalltime}. 
From the definition of $a_0$, $e_0$ and $\beta_t$ we note the simplification
\begin{equation}\label{eq:charactproperty}
\frac{\I v\E^{-\kappa t} a_0^3(k) u \beta _t}{e_0^2(k)\tau^{1/8}}=\frac{\I uk}{\tau^{1/8}}.
\end{equation}
Similarly we find that as $\tau$ tends to zero
\begin{align} \label{eq:BratioAsymp}
\frac{\widehat{B}\left(u^*_{\tau}\right)v\E^{-\kappa t}}{1-2\beta_t\widehat{B}\left(u^*_{\tau}\right)}
&=\frac{v\E^{-\kappa t}\tau^{-1/4}\widehat{B}\left(u^*_{\tau}\right)}{e_\tau} 
= \frac{v\E^{-\kappa t}\tau^{-1/4}\left(a_0^2/2+a_0 a_1 \tau^{1/4}+ r\tau^{1/2}+\mathcal{O}(\tau^{3/4})\right)}{e_0+e_1\tau^{1/4}+e_2\tau^{1/2}+\mathcal{O}(\tau^{3/4})} \\ \nonumber
&=\frac{a_0^2 v\E^{-\kappa t}}{2 e_0}\tau^{-1/4}+ v\E^{-\kappa t}\left(\frac{a_0 a_1}{e_0}-\frac{a_0^2 e_1}{2 e_0^2}\right)+\psi_4 \tau^{1/4}
+\mathcal{O}(\tau^{1/2}).
\end{align}
Again we use $e_{\tau}$ defined in~\eqref{eq:etausmalltime} to re-write the term
$$
\exp\left[\frac{2\kappa\theta}{\xi^2}
\log\left(\frac{1-2\beta_t\widehat{B}\left(u^*_{\tau}\right)}{1-2\beta_t\widehat{B}\left(u^*_{\tau}+\I u\tau^{3/8}\right)}\right)
\right]
=\left(1-\frac{2\beta_t\left(\widehat{B}\left(u^*_{\tau}+\I u\tau^{3/8}\right)-\widehat{B}\left(u^*_{\tau}\right)\right)}{e_\tau\tau^{1/4}}\right)^{-2\kappa\theta/\xi^2},
$$
and then use the asymptotics in~\eqref{eq:BIasymp} and Lemma~\ref{Lemma:etausmalltime} 
to find that for small~$\tau$
\begin{align} \label{eq:BLogAsymp}
\left(1-\frac{2\beta_t\left(\widehat{B}\left(u^*_{\tau}+\I u\tau^{3/8}\right)
-\widehat{B}\left(u^*_{\tau}\right)\right)}{e_\tau\tau^{1/4}}\right)^{-\frac{2\kappa\theta}{\xi^2}}
 & = \left(1+\frac{2\I a_0\beta_t u}{e_0}\tau^{1/8}-\frac{4 a_0^2  u^2 \beta _t^2}{e_0^2}\tau^{1/4}+\mathcal{O}(\tau^{3/8})\right)^{2\kappa\theta/\xi^2} \nonumber\\ 
 & = 1+\frac{4 \I\kappa\theta a_0\beta_t u}{e_0 \xi ^2}\tau^{1/8}-\frac{4\kappa\theta a_0^2\beta _t^2u^2
 \left(2\kappa\theta +\xi^2\right)}{\xi ^4 e_0^2}\tau^{1/4} +\mathcal{O}(\tau^{3/8}).
\end{align}
Using~\eqref{eq:phiasympsmalltime} with definition~\eqref{eq:LambdaTau} and~\eqref{eq:RenLamAB}, property~\eqref{eq:charactproperty} and the asymptotics in~\eqref{eq:AIAsymp},~\eqref{eq:BratioIAsymp},~\eqref{eq:BratioAsymp} and~\eqref{eq:BLogAsymp} we finally calculate the characteristic function for small $\tau$ as
\begin{align*}
\Phi_{k,\tau}(u)
=\exp\left(-\frac{\zeta^2u^2}{2}+(\I u \psi_0 +\I u^3 \psi_1)\tau^{1/8}
+(\psi_2 u^2 +\psi_3 u^4)\tau^{1/4}+\mathcal{O}(\tau^{3/8})\right) \\ \nonumber
\left(1+\frac{4 \I\kappa\theta a_0\beta_t u}{e_0 \xi ^2}\tau^{1/8}
 - \frac{4\kappa\theta a_0^2\beta _t^2u^2
 \left(2\kappa\theta +\xi^2\right)}{\xi ^4 e_0^2}\tau^{1/4} +\mathcal{O}(\tau^{3/8})\right),
\end{align*}
with $\psi_0,\ldots,\psi_3$ defined in~\eqref{eq:sigmarsmalltime},~\eqref{eq:psismalltime}, 
and so the lemma follows from the following decomposition 
\begin{align*}
&\Phi_{k,\tau}(u)=\exp\left(-\frac{\zeta^2u^2}{2}\right)\Big\{1+\I  \left(u 
\left(\psi _0+\frac{4\kappa\theta a_0 \beta _t}{e_0 \xi ^2}\right)+u^3 \psi _1\right)\tau^{1/8} \Big. \\ \nonumber
& \Big. +\left[
\left(\psi _2-\frac{\psi _0^2}{2} - \frac{4\kappa\theta a_0^2\beta _t^2
 \left(2\kappa\theta +\xi^2\right)}{\xi ^4 e_0^2} - \frac{4\kappa\theta a_0 \beta_t}{e_0\xi^2}\psi_0\right)u^2 
+ \left(\psi_3- \psi_0 \psi_1 - \psi_1\frac{4\kappa\theta a_0\beta_t}{e_0\xi^2}\right)u^4
-\frac{u^6 \psi _1^2}{2}\right]\tau^{\frac{1}{4}} 
+\mathcal{O}(\tau^{\frac{3}{8}})\Big\}.
\end{align*}
\end{proof}

The following technical lemma will be needed in Section~\ref{sec:optpricefwdsmile} where it will be used to give the leading order exponential decay of out-of-the-money forward-start options as $\tau$ tends to zero.

\begin{lemma} \label{Lemma:U*Asymptotics}
The following expansion holds  for all $k\in\mathbb{R}^*$ as $\tau$ tends to zero:
$$
\E^{-ku^*_{\tau}/\sqrt{\tau}+\Lambda_{\tau}^{(t)}(u^*_\tau)/\sqrt{\tau}}
=
\E^{-\Lambda^*(k)/\sqrt{\tau}+c_0(k)/\tau^{1/4}+c_1(k)}\tau^{-\kappa\theta/(2\xi^2)}
c_2(k)\left(1+z_1(k)\tau^{1/4}+\mathcal{O}(\tau^{1/2})\right),
$$
where $c_0$, $c_1$ and $c_2$ are defined in~\eqref{eq:csmalltime}, $\Lambda^*$ is characterised explicitly in Lemma~\ref{lemma:U*Characterisation} and $z_1$ is given in~\eqref{eq:zpsmalltime}.
\end{lemma}

\begin{proof}
We use the asymptotics in Lemma~\ref{Lemma:ustartausmalltime} and the characterisation of $\Lambda^*$ in Lemma~\ref{lemma:U*Characterisation} to write for small $\tau$ (we drop the $k$-dependence)
\begin{align}\label{eq:u*tausmalltime}
\exp\left(-ku^*_{\tau}/\sqrt{\tau}\right)&=\exp\left(-a_0 k/\sqrt{\tau}-a_1 k/\tau^{1/4}-a_2 k\right)\left(1-a_3 k\tau^{1/4}+\mathcal{O}(\tau^{1/2})\right) \\ \nonumber
&=\exp\left(-\Lambda^*(k)/\sqrt{\tau}-a_1 k/\tau^{1/4}-a_2 k\right)\left(1-a_3 k\tau^{1/4}+\mathcal{O}(\tau^{1/2})\right).
\end{align}
Using the Heston forward lmgf definition in~\eqref{eq:LambdaTau} and~\eqref{eq:RenLamAB} we can write
\begin{align}\label{eq:elambdatau}
\exp\left(\Lambda_{\tau}^{(t)}(u^*_\tau)/\sqrt{\tau}\right)=\exp\left(\widehat{A}(u^*_{\tau})+\frac{\widehat{B}(u^*_{\tau})v\E^{-\kappa t}}{1-2\beta_t\widehat{B}(u^*_{\tau})}-\frac{2\kappa\theta}{\xi^2}\log(1-2\beta_t\widehat{B}(u^*_{\tau}))\right).
\end{align}
Using the definition of $e_{\tau}$ in~\eqref{eq:etausmalltime} and the asymptotics in Lemma~\ref{Lemma:etausmalltime} we find that for small $\tau$
\begin{align}\label{eq:logasympsmalltime}
\exp\left(-\frac{2\kappa\theta}{\xi^2}\log(1-2\beta_t\widehat{B}(u^*_{\tau}))\right)
&=\tau^{-\kappa\theta/(2\xi^2)}e_{\tau}^{-2\kappa\theta/\xi^2}
=\tau^{-\kappa\theta/(2\xi^2)}\left(e_0+e_1\tau^{1/4}+\mathcal{O}(\tau^{1/2})\right)^{-2\kappa\theta/\xi^2} \\ \nonumber
&=\tau^{-\kappa\theta/(2\xi^2)}e_0^{-2\kappa\theta/\xi^2}
\left(1-\frac{2\kappa\theta e_1}{\xi^2 e_0}\tau^{1/4}+\mathcal{O}(\tau^{1/2})\right).
\end{align}
The the lemma follows after using~\eqref{eq:u*tausmalltime} and ~\eqref{eq:elambdatau}, 
the asymptotics in~\eqref{eq:logasympsmalltime},~\eqref{eq:BratioAsymp} and~\eqref{eq:AIAsymp} 
and the simplification
$c_0(k)=v\E^{-\kappa t}/(2 e_0(k)\beta_t)-a_1(k) k=2|a_1(k) k|$.
\end{proof}

We now use the characteristic function expansion in Lemma~\ref{Lemma:charactfuncasympsmalltime} and 
Fourier transform methods to derive the asymptotics for the expectation (under the measure~\eqref{eq:MeasureChangeSmallTime}) 
of the modified payoff on the re-scaled forward price process. 
This lemma will be critical for the analysis in 
Section~\ref{sec:optpricefwdsmile}.

\begin{lemma} \label{Lemma:fourierasympsmalltime}
The following expansion holds for all $k\in\mathbb{R}^*$ as $\tau$ tends to zero:
\begin{align*}
&\mathbb{E}^{\mathbb{Q}_{\tau,k}}\left[\E^{-u^*_{\tau}Z_{\tau,k}/\tau^{3/8}}\left(\E^{Z_{\tau,k}\tau^{1/8}}-1\right)^+\right]
\ind_{\{k>0\}}
+
\mathbb{E}^{\mathbb{Q}_{\tau,k}}\left[\E^{-u^*_{\tau}Z_{\tau,k}/\tau^{3/8}}\left(1-\E^{Z_{\tau,k}\tau^{1/8}}\right)^+\right]
\ind_{\{k<0\}} \\
&=
\frac{\tau^{7/8}\beta_t}{\zeta(k) \sqrt{2\pi}}\left(1+p_1(k)\tau^{1/4}+o\left(\tau^{1/4}\right)\right),
\end{align*}
where $\zeta$ is defined in~\eqref{eq:sigmarsmalltime}, $p_1$ in~\eqref{eq:zpsmalltime} and $\beta_t$ in~\eqref{eq:DGammaBeta}.
\end{lemma}
\begin{proof}
We first consider $k>0$ and drop the $k$-dependence for the functions below. 
We denote the Fourier transform $\mathcal{F}$ by
$
(\mathcal{F}f)(u):=\int_{-\infty}^{\infty}\E^{\I u x}f(x) \D x, 
$
for all $f\in L^2$, $u\in\mathbb{R}$.
The Fourier transform of the payoff
$\E^{-u^*_{\tau}Z_{\tau,k}/\tau^{3/8}}\left(\E^{Z_{\tau,k}\tau^{1/8}}-1\right)^+$
is given by
\begin{align*}
\int_0^{\infty}\E^{-u^*_{\tau}z/\tau^{3/8}}\left(\E^{z\tau^{1/8}}-1\right)\E^{\I u z}\D z
&=\left[\frac{\E^{z\left(\I u-u^*_{\tau}/\tau^{3/8}+\tau^{1/8}\right)}}
{\left(\I u-u^*_{\tau}/\tau^{3/8}+\tau^{1/8}\right)}\right]_{0}^{\infty}
-\left[\frac{\E^{z\left(\I u-u^*_{\tau}/\tau^{3/8}\right)}}
{\left(\I u-u^*_{\tau}/\tau^{3/8}\right)}\right]_{0}^{\infty} \\
&=\frac{1}{\left(\I u-u^*_{\tau}/\tau^{3/8}\right)}-\frac{1}{\left(\I u-u^*_{\tau}/\tau^{3/8}+\tau^{1/8}\right)} \\ 
&=\frac{\tau ^{7/8}}{\left(u_{\tau }^*-\I \tau ^{3/8} u\right) \left(u_{\tau }^*-\sqrt{\tau }-\I \tau ^{3/8} u\right)},
\end{align*}
if $u^*_{\tau}>\max(\tau^{1/2},0)=\tau^{1/2}$, which holds for $\tau$ small enough since 
$u^*_{\tau}$ converges to $a_0>0$ by Lemma~\ref{Lemma:ustartausmalltime}.
Due to Remark~\ref{rem:weakconsmalltau} $Z_{\tau}$ converges weakly to a Gaussian random variable and since
the Gaussian density and the modified payoff are in $L^2$ we can use Parseval's Theorem~\cite[Page 48,Theorem 13E]{G70} for small enough $\tau$ to write
\begin{equation}\label{eq:parsevalnonsteepsmalltime}
\mathbb{E}^{\mathbb{Q}_{\tau,k}}\left(\E^{-u^*_{\tau}Z_{\tau,k}/\tau^{3/8}}\left(\E^{Z_{\tau,k}\tau^{1/8}}-1\right)^+\right)
=\frac{1}{2\pi}\int_{-\infty}^{\infty}\frac{\tau ^{7/8}\Phi_{\tau,k}(u)}{\left(u_{\tau }^*+\I \tau ^{3/8} u\right) \left(u_{\tau }^*-\sqrt{\tau }+\I \tau ^{3/8} u\right)} \D u,
\end{equation}
where we have used that 
\begin{align*}
\overline{\frac{\tau ^{7/8}}{\left(u_{\tau }^*-\I \tau ^{3/8} u\right) \left(u_{\tau }^*-\sqrt{\tau }-\I \tau ^{3/8} u\right)}}=\frac{\tau ^{7/8}}{\left(u_{\tau }^*+\I \tau ^{3/8} u\right) \left(u_{\tau }^*-\sqrt{\tau }+\I \tau ^{3/8} u\right)},
\end{align*}
with $\overline{a}$ denoting the complex conjugate for $a\in\mathbb{C}$. 
Using the asymptotics of $u^*_{\tau}$ given in Lemma~\ref{Lemma:ustartausmalltime} we can Taylor expand for small $\tau$ to find that
\begin{equation}\label{eq:conjasympsmalltime}
\frac{\tau ^{7/8}}{\left(u_{\tau }^*+\I \tau ^{3/8} u\right) \left(u_{\tau }^*-\sqrt{\tau }+\I \tau ^{3/8} u\right)}
=\frac{\tau^{7/8}}{a_0^2+2a_0a_1\tau^{1/4} +\mathcal{O}(\tau^{3/8})}
=\frac{\tau^{7/8}}{a_0^2}\left(1-\frac{2 a_1}{a_0}\tau^{1/4}+\mathcal{O}\left(\tau^{3/8}\right)\right).
\end{equation}
Finally combining~\eqref{eq:conjasympsmalltime} and the asymptotics of the characteristic function derived in  Lemma~\ref{Lemma:charactfuncasympsmalltime} with~\eqref{eq:parsevalnonsteepsmalltime} we find that for small $\tau$
\begin{align}
\nonumber &\frac{1}{2\pi}\int_{-\infty}^{\infty} \frac{\tau ^{7/8}\Phi_{\tau,k}(u)}{\left(u_{\tau }^*+\I \tau ^{3/8} u\right) \left(u_{\tau }^*-\sqrt{\tau }+\I \tau ^{3/8} u\right)} \D u \\ \nonumber
&=\frac{\tau^{7/8}}{a_0^2 2\pi}\int_{-\infty}^{\infty}\E^{-\frac{\zeta^2u^2}{2}}\left(1+\phi_1(u,k)\tau^{1/8}+\left(\phi_2(u,k)-\frac{2a_1}{a_0}\right)\tau^{1/4}+\mathcal{O}\left(\tau^{3/8}\right)\right)\D u \\ \nonumber
&=\frac{\tau^{7/8}}{a_0^2 2\pi} \int_{-\infty}^{\infty}\E^{-\frac{\zeta^2u^2}{2}}\left(1+\left(u^2\phi_2^a+u^4\phi_2^b+u^6\phi_2^c-\frac{2a_1}{a_0}\right)\tau^{1/4}+\mathcal{O}\left(\tau^{3/8}\right)\right) \D u,
\end{align}
where in the last line we have used that
$
 \int_{-\infty}^{\infty}\E^{-\frac{\zeta^2u^2}{2}}\phi_1(u,k) \D u =0,
$
since $\phi_1$ is an odd power of $u$. 
The result then follows by using simple moment formulae of the normal distribution. 
Fix now $k<0$. The Fourier transform of the payoff
$\E^{-u^*_{\tau}Z_{\tau,k}/\tau^{3/8}}\left(1-\E^{Z_{\tau,k}\tau^{1/8}}\right)^+$
is given by
\begin{align*}
\int_{-\infty}^0\E^{-u^*_{\tau}z/\tau^{3/8}}\left(1-\E^{z\tau^{1/8}}\right)\E^{\I u z}\D z
&=\left[\frac{\E^{z\left(\I u-u^*_{\tau}/\tau^{3/8}\right)}}
{\left(\I u-u^*_{\tau}/\tau^{3/8}\right)}\right]_{-\infty}^{0}-\left[\frac{\E^{z\left(\I u-u^*_{\tau}/\tau^{3/8}+\tau^{1/8}\right)}}
{\left(\I u-u^*_{\tau}/\tau^{3/8}+\tau^{1/8}\right)}\right]_{-\infty}^{0} \\
&=\frac{1}{\left(\I u-u^*_{\tau}/\tau^{3/8}\right)}-\frac{1}{\left(\I u-u^*_{\tau}/\tau^{3/8}+\tau^{1/8}\right)} \\ 
&=\frac{\tau ^{7/8}}{\left(u_{\tau }^*-\I \tau ^{3/8} u\right) \left(u_{\tau }^*-\sqrt{\tau }-\I \tau ^{3/8} u\right)},
\end{align*}
if $u^*_{\tau}<\min(\tau^{1/2},0)=0$, which holds for $\tau$ small enough since 
$u^*_{\tau}$ converges to $a_0<0$ by Lemma~\ref{Lemma:ustartausmalltime}.
The rest of the proof is analogous to $k<0$ above and we omit it for brevity.
\end{proof}

\begin{remark}
We have chosen to specify the remainder in the form $o(1/{\tau^{1/4}})$ instead of $\mathcal{O}(1/\tau^{3/8})$ since it can actually be shown that the term $\mathcal{O}(1/\tau^{3/8})$ is zero by extending the results in Lemma~\ref{Lemma:charactfuncasympsmalltime} and the next non-trivial term is $\mathcal{O}(1/\tau^{1/2})$. 
For brevity we omit this analysis.
\end{remark}

\subsection{Option price and forward smile asymptotics}\label{sec:optpricefwdsmile}

In this section we finally put all the pieces together from Sections~\ref{sec:hestfwdtime} - \ref{sec: Charactfuncasymp} and prove Theorems~\ref{theorem:fwdstartoptionssmalltime} and~\ref{theorem:fwdsmilesmalltime}.

\begin{proof}[Proof of Theorem~\ref{theorem:fwdstartoptionssmalltime}]
We use the time-dependent change of measure defined in~\eqref{eq:MeasureChangeSmallTime} 
to write forward-start call option prices for all $k>0$ as
\begin{eqnarray}
\nonumber\mathbb{E}\left[\left(\E^{X^{( t)}_{\tau}}-\E^k\right)^+\right] &=&\E^{\Lambda^{(t)}_{\tau}\left(u^*_{\tau}\right)/\sqrt{\tau}}\mathbb{E}^{\mathbb{Q}_{k,\tau}}\left[\E^{- u^*_{\tau}X_{\tau}^{(t)}/\sqrt{\tau}}\left(\E^{X^{( t)}_{\tau}}-\E^k\right)^+\right] \\ \nonumber
&=&\E^{-\frac{ku^*_{\tau}-\Lambda^{(t)}_{\tau}\left(u^*_{\tau}\right)}{\sqrt{\tau}}}\mathbb{E}^{\mathbb{Q}_{k,\tau}}\left[\E^{-\frac{u^*_{\tau}}{\sqrt{\tau}}\left(X_{\tau}^{(t)}- k\right)}\left(\E^{X^{( t)}_{\tau}}-\E^k\right)^+\right] \\ \nonumber
&=&\E^{-\frac{ku^*_{\tau}-\Lambda^{(t)}_{\tau}\left(u^*_{\tau}\right)}{\sqrt{\tau}}}\E^{k}\mathbb{E}^{\mathbb{Q}_{k,\tau}}\left[\E^{-\frac{u^*_{\tau}Z_{\tau,k}}{\tau^{3/8}}}\left(\E^{Z_{\tau,k}\tau^{1/8}}-1\right)^+\right],
\end{eqnarray}
with $Z_{\tau,k}$ defined on page~\pageref{eq:DefZSmallTime}. 
A similar result holds for forward-start put option prices for all $k<0$.
The theorem then follows  by applying Lemma~\ref{Lemma:U*Asymptotics} and Lemma~\ref{Lemma:fourierasympsmalltime}
and using put-call parity since in the Heston model $(\E^{X_t})_{t\geq0}$ is a true martingale~\cite[Proposition 2.5]{AP07}. 
\end{proof}

\begin{proof}[Proof of Theorem~\ref{theorem:fwdsmilesmalltime}]
The general machinery to translate option price asymptotics into implied volatility asymptotics
has been fully developed by Gao and Lee~\cite{GL11}. 
We simply outline the main steps here.
Assume the following ansatz for the forward implied volatility as $\tau$ tends to zero:
$$
\sigma_{t,\tau}^2(k)=\frac{v_0(k,t)}{\sqrt{\tau }}+\frac{v_1(k,t)}{\tau^{1/4}}+v_2(k,t)+v_3(k,t)\tau^{1/4}+o(\tau^{1/4}).
$$
Substituting this ansatz into the BSM asymptotics in Lemma~\ref{lemma:BSAsympSmallTau} we then obtain
$$
\exp \left(-\frac{k^2}{2 \sqrt{\tau } v_0}+\frac{k^2 v_1}{2 \tau^{1/4} v_0^2}-\frac{k^2 \left(v_1^2-v_0 v_2\right)}{2 v_0^3}+\frac{k}{2}\right)
\frac{\tau ^{3/4} v_0^{3/2}}{\sqrt{2 \pi } k^2}
\left[1+\left(\frac{k^2 \left(v_1^3-2 v_0 v_1 v_2+v_0^2 v_3\right)}{2 v_0^4}+\frac{3 v_1}{2 v_0}\right)\tau^{1/4}+o(\tau^{1/4})\right].
$$
Equating orders with Theorem~\ref{theorem:fwdstartoptionssmalltime} we solve for $v_0$ and $v_1$, but we can only solve for higher order terms if $\tau^{3/4}=\tau^{(7/8-\theta\kappa/(2\xi^2))}$ or $4\kappa\theta=\xi^2$.
\end{proof}


\end{document}